\documentclass[12pt]{article}

\usepackage{amsmath}
\usepackage{amssymb}
\usepackage{amsthm}
\usepackage{url}
\usepackage{changepage}
\usepackage{microtype}
\usepackage{tikz}
\usetikzlibrary{arrows.meta,decorations.pathmorphing,positioning}
\usepackage{centernot}
\usepackage{turnstile}
\usepackage{xcolor}
\usepackage{etoolbox}
\usepackage{needspace}
\usepackage[hidelinks,unicode]{hyperref}
\hypersetup{pdftitle={Toward a Characterization of Simulation Between Arithmetic Theories},pdfauthor={Hunter Monroe},pdfsubject={ECCC TR26-062, Version 3}}
\AtBeginEnvironment{theorem}{\Needspace{5\baselineskip}}
\AtBeginEnvironment{lemma}{\Needspace{5\baselineskip}}
\AtBeginEnvironment{corollary}{\Needspace{5\baselineskip}}
\AtBeginEnvironment{proposition}{\Needspace{5\baselineskip}}
\AtBeginEnvironment{conjecture}{\Needspace{5\baselineskip}}
\AtBeginEnvironment{definition}{\Needspace{5\baselineskip}}

\AtBeginEnvironment{question}{\Needspace{5\baselineskip}}

\definecolor{hardred}{RGB}{180,30,30}
\definecolor{oxblue}{RGB}{0,33,71}
\definecolor{camblue}{RGB}{163,193,173}
\definecolor{midgray}{RGB}{210,210,210}

\theoremstyle{plain}
\newtheorem{theorem}{Theorem}[section]
\newtheorem{lemma}[theorem]{Lemma}
\newtheorem{corollary}[theorem]{Corollary}
\newtheorem{conjecture}[theorem]{Conjecture}
\newtheorem{proposition}[theorem]{Proposition}

\theoremstyle{definition}
\newtheorem{definition}[theorem]{Definition}

\title{\vspace{-3.5em}
Toward a Characterization of Simulation Between Arithmetic Theories
}
\author{Hunter Monroe}
\date{}
\begin{document}
\maketitle
\vspace{-2.5em}
\begin{abstract}
We study when a sound arithmetic theory $\mathcal S{\supseteq}S^1_2$ with polynomial-time decidable axioms efficiently proves the bounded consistency statements $Con_{\mathcal S{+}\phi}(n)$ for a true sentence $\phi$. Our main structural result characterizes the nonexistence of a length-optimal propositional proof system: this is equivalent to every such $\mathcal S$ failing to simulate $S^1_2{+}\mathsf{Exp}{+}\phi_{BB}(k)$ for all sufficiently large $k$, where $\phi_{BB}(k)$ asserts the exact value of the $k$-state Busy Beaver function and $\mathsf{Exp}$ asserts total exponentiation. We obtain this characterization by combining the Kraj\'i\v{c}ek--Pudl\'ak correspondence with our transfer theorem: for each fixed $\mathcal S$, the existence of a true extension not simulated by $\mathcal S$ is equivalent to this eventual Busy Beaver non-simulation property. A finite version transfers individual lower bounds to Busy Beaver instances with explicit parameter changes and proof-length bounds. For finitely axiomatized sequential $\mathcal S$, we also prove that $S^1_2{\vdash}Con_{\mathcal S}{\rightarrow}Con_{\mathcal S{+}\phi}$ yields polynomial-size $\mathcal S$-proofs of $Con_{\mathcal S{+}\phi}(n)$. We refute our earlier conjecture that $EA{\vdash}Con_{\mathcal S}{\rightarrow}Con_{\mathcal S{+}\phi}$ characterizes simulation.

In this finitely axiomatized sequential setting, a new Kolmogorov Blindness conjecture predicts hardness of random tautologies: true logarithmic-deficiency randomness axioms $x{\in}R^{\log}$ yield bounded-consistency statements requiring exponentially long proofs, at a rate depending on the fixed proof encoding, beyond an onset common to all strings of each sufficiently large length. It further predicts this hardness after adjoining a true randomness axiom $y{\in}R^{\log}$ with $|x|{=}|y|$, provided $x$ remains random conditional on $y$. For fixed $\Pi_1$ axiom families with computable proof-length bounds and a common computable onset, we prove a quantitative equivalence between dense hardness with polynomially small exceptional fractions and hardness above a corresponding Kolmogorov-complexity threshold.
\end{abstract}

\section{Introduction}
\emph{``\ldots the essence of the big open problems in complexity theory could be logical, rather than combinatorial''} --- Pudl\'ak \cite[p.~581]{PudlakLogicalFoundations}

\Needspace{5\baselineskip}
Let $\mathcal S{\supseteq}S^1_2$ be a sound arithmetic theory with polynomial-time decidable axioms, and let $\phi$ be a true sentence. We study when $\mathcal S$ efficiently proves the bounded consistency statements $Con_{\mathcal S{+}\phi}(n)$, where $Con_{\mathcal S{+}\phi}(n)$ asserts that $\mathcal S{+}\phi$ has no proof of $0{=}1$ within $n$ bits under a fixed proof encoding. The question of how long these consistency proofs must be goes back to G\"odel, as reported by Kreisel (1967, p.~241); see Pudl\'ak~\cite[p.~628]{PudlakLengthsOfProofs}.

The finite-consistency comparison asks whether $\mathcal S$ has proofs of $Con_{\mathcal S{+}\phi}(n)$ of size polynomial in $n$~(Kraj\'i\v{c}ek and Pudl\'ak~\cite{Krajicek}); we abbreviate this by saying that $\mathcal S$ simulates $\mathcal S{+}\phi$.  The literature provides sufficient conditions for $\mathcal S$ to simulate $\mathcal S{+}\phi$: Hrube\v{s}'s self-reference construction and Je\v{r}\'abek's interpretation argument give true independent extensions with short bounded-consistency proofs under their respective hypotheses~(Pudl\'ak~\cite[Theorem~3.4 and Lemma~3.6]{Pudlak2017Author}; see Theorems~\ref{thm:hrubes} and~\ref{thm:interpretability-implies-simulation}).

On the other hand, if no length-optimal propositional proof system exists, then for each $\mathcal S$ as above there is a true sentence $\phi$ such that $\mathcal S\centernot{\sststile{}{n^c}}Con_{\mathcal S{+}\phi}(n)$ for every constant $c{>}0$; see Kraj\'i\v{c}ek and Pudl\'ak~\cite[Theorem~2.1]{Krajicek} and Corollary~\ref{cor:optimal-characterization} below. Thus the problem is to identify a criterion for $\phi$ that governs simulation and non-simulation. We formulate this as follows.

\begin{quote}
\noindent\textbf{Characterization Problem.}
For a sound theory $\mathcal S{\supseteq}S^1_2$ with polynomial-time decidable axioms and a true sentence $\phi$, determine when $\mathcal S$ has polynomial-size proofs of $Con_{\mathcal S{+}\phi}(n)$ and when no polynomial bounds their lengths.
\end{quote}

The paper's main structural result is an exact Busy Beaver characterization of the nonexistence of a length-optimal propositional proof system (Corollary~\ref{cor:optimal-characterization}): such a system does not exist if and only if every sound $\mathcal S{\supseteq}S^1_2$ with polynomial-time decidable axioms fails to simulate $S^1_2{+}\mathsf{Exp}{+}\phi_{BB}(k)$ for all sufficiently large $k$. This combines the Kraj\'i\v{c}ek--Pudl\'ak correspondence~\cite[Theorem~2.1]{Krajicek} with our transfer theorem. Under this equivalent no-optimality hypothesis, the propositional encodings of these Busy Beaver consistency statements give a canonical collection of tautology families: every Cook--Reckhow proof system has no polynomial bound on proof lengths as $n$ varies in each family with sufficiently large $k$, with the threshold depending on the proof system. Classification of individual true extensions remains open.

On the positive side, we give a sufficient condition for simulation requiring only relative consistency provable in $S^1_2$. If $S^1_2$ proves $Con_{\mathcal S}{\rightarrow}Con_{\mathcal S{+}\phi}$, Parikh's theorem~(Buss~\cite{Buss1986bounded}) gives a polynomial bounded-consistency transfer; combined with short self-consistency proofs, this yields simulation (Theorem~\ref{thm:relative-consistency-implies-simulation}). Finitely axiomatized sequential theories have the required short self-consistency proofs~(Pudl\'ak~\cite{Pudlak1986length}). This includes the verified interpretation mechanism of Je\v{r}\'abek, as presented by Pudl\'ak~\cite[Lemma~3.6]{Pudlak2017Author}. On the hardness side, Pudl\'ak~\cite[Problem~1, p.~193]{Pudlak1986length} conjectures that simulation fails for $\phi{=}Con_{\mathcal S}$.

The Busy Beaver transfer identifies exact Busy Beaver values as canonical witnesses to non-simulation. For each sound $\mathcal S{\supseteq}S^1_2$ with polynomial-time decidable axioms, a hard true extension exists if and only if $\mathcal S$ fails to simulate $S^1_2{+}\mathsf{Exp}{+}\phi_{BB}(k)$ for all sufficiently large $k$ (Theorem~\ref{thm:busy-beaver-transfer}). Here $\phi_{BB}(k)$ specifies the maximum halting time of a halting $k$-state Turing machine~(Rad\'o~\cite{Rado1962}). Exponentiation is used inside the witness theory; $\mathcal S$ need only extend $S^1_2$. A finite version transfers an individual lower bound to a Busy Beaver instance, and conversely to a suitable true extension, with explicit changes of parameter and proof length (Theorem~\ref{thm:finite-busy-beaver-transfer}). For any class of such theories each having a hard true extension, a nonconstructive diagonal choice also gives one Busy Beaver consistency family hard throughout that class (Corollary~\ref{cor:diagonal-busy-beaver-family}); no hypothesis is imposed on theories outside the class. Alternative canonical witnesses over $\mathsf{PA}$ are exact prefixes of Chaitin's halting probability and complete finite randomness lists (Proposition~\ref{prop:alternative-canonical-witnesses}).

We propose Kolmogorov randomness as an obstruction to simulation: our Kolmogorov Blindness conjecture predicts exponential lower bounds for bounded-consistency proofs associated with sufficiently long true randomness axioms (Conjecture~\ref{conj:kolmogorov-blindness}). It further asserts that the bounds hold after adjoining another true randomness axiom of the same string length, provided the target string remains random conditional on the added string. For sound finitely axiomatized sequential $\mathcal S{\supseteq}S^1_2$, it implies Strong KH, requiring non-simulation for every sufficiently long true randomness axiom, and Density KH, requiring such true hard axioms for a fraction tending to one of strings (Theorem~\ref{thm:blindness-consequences}). Its universal form also implies Conditional Factual No Access (Proposition~\ref{prop:universal-blindness-no-access}).

The density contribution makes the connection between randomness and dense hardness precise: for fixed $\Pi_1$ axiom families with computable proof-length bounds and a computable common onset, we prove a quantitative equivalence between dense hardness and hardness above a corresponding randomness threshold (Section~\ref{sec:effective-density}). One direction is counting: the logarithmic-deficiency randomness set $R^{\log}$ defined in Section~\ref{sec:higher-relative-consistency} has density tending to one, so hardness for every sufficiently long member gives dense hardness. Conversely, if the fraction of exceptions to truth or any prescribed bound is $O(m^{-a})$, every sufficiently long string with $K_U(x){\ge}m{-}a\log_2m{+}c$, for a suitable constant $c$, gives a true instance satisfying every bound after the common onset. Exceptions can be effectively enumerated under the stated hypotheses; their sparsity remains a premise. The pairwise version uses joint complexity and requires both bounds to hold when the other axiom is adjoined. Density is measured among indices, with the same strings satisfying the whole prescribed tail.

Together, these results constrain how fine-grained and how coarse-grained the ordering of theories by simulation can be. Verified relative consistency prevents a separation by simulation, while any failure to simulate a true extension must also appear among canonical Busy Beaver theories. The density equivalence further constrains which random indices can be exceptions to prescribed hardness bounds.

The first version of this paper proposed that, for sound finitely axiomatized sequential $\mathcal S{\supseteq}S^1_2$, simulation of a true extension $\mathcal S{+}\phi$ is equivalent to $EA{\vdash}Con_{\mathcal S}{\rightarrow}Con_{\mathcal S{+}\phi}$. Section~\ref{sec:desiderata} refutes its necessary direction, Higher Relative Consistency (HRC), using Hrube\v{s}'s self-reference construction. Replacing $EA$ by $S^1_2$ does not repair the characterization.

A standard folklore argument shows that a consistent $\mathcal T{\supseteq}EA$ with polynomial-time decidable axioms cannot prove the uniform assertion of its own non-simulation of a fixed extension: that assertion implies $Con_{\mathcal T}$. Section~\ref{sec:provability} gives the precise statement and discusses why randomness relative to a stronger theory suggests an obstacle to proving hardness, without establishing such unprovability. For the corresponding propositional argument, see Pudl\'ak~\cite[Proposition~4.5 and Corollary~4.7]{Pudlak2020reflectionprinciples}.

Section~\ref{sec:preliminaries} provides preliminaries. Section~\ref{sec:unconditional-constraints} proves the relative-consistency and Busy Beaver results; Section~\ref{sec:desiderata} refutes HRC. Sections~\ref{sec:higher-relative-consistency} and~\ref{sec:effective-density} develop Kolmogorov Blindness and the effective density equivalence, respectively. Section~\ref{sec:provability} discusses provability, and Section~\ref{sec:conclusion} concludes.

\section{Preliminaries}\label{sec:preliminaries}

We work with sound arithmetical theories $\mathcal S{\supseteq}S^1_2$ with polynomial-time decidable axioms, and true sentences $\phi$, unless stated otherwise. Fix a standard arithmetization in the language of bounded arithmetic, extended by explicit definitions. Use ordinary Hilbert proofs with reusable earlier lines, equivalently a directed acyclic graph calculus with cut, and binary numerals~(Pudl\'ak~\cite[Sections~2.4--2.5]{Pudlak2017Author}). Length is the bit length of an ordinary local binary serialization: formulas are written explicitly, and numeral digits, variable identifiers, line references and headers are counted. The encoding adds no padding depending on the total proof length. Polynomial simulation is preserved by verified polynomial proof translations; a numerical exponential rate requires this particular encoding to remain fixed.

We fix a finite axiomatization of $S^1_2$~(Kraj\'i\v{c}ek~\cite[Section~2]{Krajicek2025SLog}). A theory extending $S^1_2$ either contains these axioms literally or comes with a polynomial-time proof embedding verified in $S^1_2$. Fixed proofs of the finitely many base axioms supply such an embedding. Each theory with polynomial-time decidable axioms has a \emph{polynomial-time ordinary proof predicate} $Prf_{\mathcal T}$: its fixed clocked verifier is $S^1_2$-provably equivalent to this predicate, from which $Con_{\mathcal T}$ and $Con_{\mathcal T}(n)$ are formed.

The sentence $Con_{\mathcal S}(n)$ asserts that no $\mathcal S$-proof of $0{=}1$ has length at most $n$~(Pudl\'ak~\cite{Pudlak1986length}).\footnote{Friedman's 1979 work and Pudl\'ak's subsequent results~\cite[Introduction]{Pudlak2017Author} developed the quantitative theory.} We write $\mathcal S{\sststile{}{n^c}}\phi(n)$ when, for every sufficiently large $n$, there is an $\mathcal S$-proof of $\phi(n)$ of length at most $n^c$, and $\mathcal S{\sststile{}{n^{O(1)}}}\phi(n)$ when a polynomial bounds these lengths. For a single sentence $\psi$, $\mathcal S{\sststile{}{s}}\psi$ means that a proof of length at most $s$ exists.

We abbreviate \emph{finite-consistency simulation} by \emph{simulation}: $\mathcal S$ simulates $\mathcal T$ when $\mathcal S{\sststile{}{n^{O(1)}}}Con_{\mathcal T}(n)$~(Kraj\'i\v{c}ek and Pudl\'ak~\cite{Krajicek}). The target $\mathcal T$ need not be a single-sentence extension of $\mathcal S$. For propositional systems, $P$ simulates $Q$ in the length sense if every tautology with a $Q$-proof of length $s$ has a $P$-proof of length at most $p(s)$ for one polynomial $p$. Cook and Reckhow~\cite[Definition~1.5]{CookReckhow} additionally require a polynomial-time translation for $p$-simulation. An \emph{optimal}, or \emph{length-optimal}, system simulates every Cook--Reckhow system in the length sense~(Pudl\'ak~\cite[Section~5]{Pudlak2017Author}). Write $\mathrm{RFN}(P)$ for the arithmetic sentence asserting that every formula with a $P$-proof is a tautology.

Let $\mathsf{Exp}$ assert total exponentiation: $\forall x\,\exists y\,E_2(x,y)$, where the bounded formula $E_2(x,y)$ expresses $y{=}2^x$ and $S^1_2$ verifies its elementary defining properties, including $|y|{=}x{+}1$. We treat $S^1_2{+}\mathsf{Exp}{+}\psi$ as extension by the single sentence $\mathsf{Exp}{\land}\psi$. Elementary Arithmetic $EA$ uses the common language fixed in this section, with the usual explicit definitions of $|x|$, $\lfloor x/2\rfloor$ and $x\#y$. Under this translation $EA$ proves the axioms of $S^1_2$.

A \emph{uniform feasible relative-consistency certificate} is an $S^1_2$-proof of $\forall n\,(Con_{\mathcal S}(p(n)){\rightarrow}Con_{\mathcal S{+}\phi}(n))$ for a fixed polynomial $p$. In contrast, \emph{instance-wise feasible relative consistency} means the external family $\mathcal S{\sststile{}{n^{O(1)}}}Con_{\mathcal S}(p(n)){\rightarrow}Con_{\mathcal S{+}\phi}(n)$. With short self-consistency proofs the latter is equivalent to simulation (Theorem~\ref{thm:feasible-relative-consistency}); it supplies no uniform certificate. Neither notion means $EA{+}Con_{\mathcal S}{\vdash}\phi$.

For Je\v{r}\'abek's argument, a polynomial-time interpretation includes a polynomial proof translation whose correctness is verified in $S^1_2$, including uniform translated axiom proofs for infinite presentations. The standard formalized $\Sigma^b_1$-completeness bound supplies short proofs of accepted computations~(Khaniki~\cite[Theorem~2.2]{khaniki2019}). We additionally use its syntactic computation-writing construction: for each fixed clocked verifier $V$, a polynomial-time function takes an input and an accepting computation tableau to an $S^1_2$-proof that $V$ accepts that input, of size polynomial in the input and tableau lengths. Its correctness is verified in $S^1_2$ by checking each computation step. The same convention applies to rejecting computations. This explicit construction, rather than an inference from mere existence of short proofs, is used in the proof transformations below. Replacing an axiom by a fixed derivation is another syntactic transformation verified in $S^1_2$.
\section{Relative Consistency and Busy Beaver Transfer}\label{sec:unconditional-constraints}
This section presents the paper's first two unconditional results. The first shows in certain settings that relative consistency implies feasible relative consistency and hence simulation. The second shows that any hard true $\phi$ can, in a precise sense, be replaced by one of Busy Beaver form. Together, these results isolate both a positive mechanism underlying simulation and a canonical source of hardness, constraining any possible characterization.

The literature shows that simulation is possible in special cases. In particular, when $\mathcal S$ has polynomial-size proofs of its own bounded consistency, polynomial-time interpretability implies efficient provability of $Con_{\mathcal S{+}\phi}(n)$; see Theorem~\ref{thm:interpretability-implies-simulation}.\footnote{Freund and Pakhomov~\cite[Theorem~3.6]{FreundPakhomov} prove that $\mathsf{PA}$ has polynomial-size proofs of the bounded consistency of $\mathsf{PA}{+}Con^*(\mathsf{PA})$, where $Con^*(\mathsf{PA})$ is their slow-consistency sentence. This example uses full $\mathsf{PA}$, rather than a finitely axiomatized base.} By contrast, Pudl\'ak~\cite[Problem~1, p.~193]{Pudlak1986length} conjectures that simulation fails already for $\phi{=}Con_{\mathcal S}$ in the finitely axiomatized setting (Conjecture~\ref{conj:pudlak}).\footnote{A computable jump operator assigns to each propositional proof system $P$ a proof system that $P$ cannot simulate. Khaniki~\cite[Theorem~5.11]{Khaniki2024JumpOperators} proves that the existence of such an operator is equivalent to a weakened form of Pudl\'ak's conjecture, requiring the consistency-extension lower bound only for consistent finitely axiomatized theories proving a fixed sound finitely axiomatized extension of $S^1_2$. The existence of such operators remains open.} The assertion that every sound theory $\mathcal S{\supseteq}S^1_2$ with polynomial-time decidable axioms has a true extension $\mathcal S{+}\phi$ such that $\mathcal S\centernot{\sststile{}{n^{O(1)}}}Con_{\mathcal S{+}\phi}(n)$ remains open. This universal assertion is equivalent to the nonexistence of a length-optimal propositional proof system, and hence would imply $\mathbf{NP}{\neq}\mathbf{coNP}$ (Corollary~\ref{cor:optimal-characterization}).

\subsection{A Known Sufficient Condition for Simulation}

Je\v{r}\'abek's proof-translation argument, presented by Pudl\'ak~\cite[Lemma~3.6]{Pudlak2017Author} for finitely axiomatized sequential theories, gives the following sufficient condition. Our explicit feasible-translation convention also permits infinite presentations when the required translation is supplied.

\begin{theorem}\label{thm:interpretability-implies-simulation}
(Proof-translation argument of Je\v{r}\'abek, via Pudl\'ak~\cite[Lemma~3.6]{Pudlak2017Author}) Suppose $\mathcal S{\supseteq}S^1_2$ and $\mathcal S{\sststile{}{n^{O(1)}}}Con_{\mathcal S}(n)$. If $\mathcal S$ polynomial-time interprets $\mathcal S{+}\phi$, then $\mathcal S{\sststile{}{n^{O(1)}}}Con_{\mathcal S{+}\phi}(n)$.
\end{theorem}

\begin{proof}
The interpretation translates an $(\mathcal S{+}\phi)$-refutation into an $\mathcal S$-proof of $(0{=}1)^i$ with polynomial growth. Translating a fixed proof of $0{\ne}1$ supplies an $\mathcal S$-proof of $\neg(0{=}1)^i$, hence an $\mathcal S$-refutation with polynomial growth. By our interpretation convention, $S^1_2$ verifies this construction, so for a polynomial $p$ it proves $\forall n\,(Con_{\mathcal S}(p(n)){\rightarrow}Con_{\mathcal S{+}\phi}(n))$. Instantiating this fixed proof and combining it with the assumed short proofs of $Con_{\mathcal S}(p(n))$ gives the result.
\end{proof}

The uniform certificate here is stronger than the existence of a short-proof family; Section~\ref{sec:desiderata} shows that the latter need not supply the former.

\subsection{Stronger Results on Simulation}
The interpretation argument uses a uniform bounded-consistency implication. Its instance-wise analogue has a simple characterization.

\begin{theorem}\label{thm:feasible-relative-consistency}
Let $\mathcal S{\supseteq}S^1_2$ be sound with polynomial-time decidable axioms and $\mathcal S{\sststile{}{n^{O(1)}}}Con_{\mathcal S}(n)$. For any fixed polynomial $p$ taking natural-number values, $\mathcal S$ simulates $\mathcal S{+}\phi$ if and only if $\mathcal S{\sststile{}{n^{O(1)}}}Con_{\mathcal S}(p(n)){\rightarrow}Con_{\mathcal S{+}\phi}(n)$.
\end{theorem}
\begin{proof}
Short proofs of the implication combine with short proofs of $Con_{\mathcal S}(p(n))$ to give simulation. Conversely, a proof of $Con_{\mathcal S{+}\phi}(n)$ gives the implication by propositional logic with polynomial overhead. Finitely many small values of $p(n)$ can be covered by fixed proofs of their true finite-consistency statements.
\end{proof}

The corresponding equivalence for non-simulation follows immediately.

\begin{corollary}\label{cor:simulation-failure}
Let $\mathcal S{\supseteq}S^1_2$ be sound with polynomial-time decidable axioms, and suppose $\mathcal S{\sststile{}{n^{O(1)}}}Con_{\mathcal S}(n)$. For any sentence $\phi$ and any fixed polynomial $p$ taking natural-number values, $\mathcal S$ fails to simulate $\mathcal S{+}\phi$ if and only if $\mathcal S\centernot{\sststile{}{n^{O(1)}}}Con_{\mathcal S}(p(n)){\rightarrow}Con_{\mathcal S{+}\phi}(n)$.
\end{corollary}

\begin{proof}
Negate both sides of the equivalence in Theorem~\ref{thm:feasible-relative-consistency}.
\end{proof}

The proof of Theorem~\ref{thm:interpretability-implies-simulation} uses only a relative-consistency implication that $S^1_2$ proves uniformly in $n$. The next theorem shows that such an implication follows whenever $S^1_2$ proves the unbounded implication $Con_{\mathcal S}{\rightarrow}Con_{\mathcal S{+}\phi}$; no interpretation is needed. Note that the hypothesis $\mathcal S{\sststile{}{n^{O(1)}}}Con_{\mathcal S}(n)$ holds whenever $\mathcal S$ is finitely axiomatized and sequential~(Pudl\'ak~\cite[Theorem~5.4 and p.~189]{Pudlak1986length}).

\begin{theorem}\label{thm:relative-consistency-implies-simulation}
Let $\mathcal S{\supseteq}S^1_2$ have a polynomial-time ordinary proof predicate, and suppose $\mathcal S{\sststile{}{n^{O(1)}}}Con_{\mathcal S}(n)$. Let $\phi$ be a sentence. If $S^1_2{\vdash}Con_{\mathcal S}{\rightarrow}Con_{\mathcal S{+}\phi}$, then there is a polynomial $q$ with $q(n){\ge}n$ such that $S^1_2\vdash\forall n\,\bigl(Con_{\mathcal S}(q(n))\rightarrow Con_{\mathcal S{+}\phi}(n)\bigr)$. Its instances have polynomial-size $\mathcal S$-proofs, and $\mathcal S$ simulates $\mathcal S{+}\phi$.
\end{theorem}

\begin{proof}
Write $\bot$ for $\ulcorner0{=}1\urcorner$. By the conventions of Section~\ref{sec:preliminaries}, $Prf_{\mathcal S}$ and $Prf_{\mathcal S{+}\phi}$ are computed by clocked polynomial-time verifiers whose computations are $S^1_2$-provably equivalent to them, so each may be taken to be a bounded formula. Classical logic turns the hypothesis into $S^1_2\vdash\forall\pi\,\exists\rho\,\bigl(\neg Prf_{\mathcal S{+}\phi}(\pi,\bot)\vee Prf_{\mathcal S}(\rho,\bot)\bigr)$: if $\pi$ is not an $(\mathcal S{+}\phi)$-refutation, any $\rho$ will do, and otherwise $\neg Con_{\mathcal S{+}\phi}$ holds, so the hypothesis supplies an $\mathcal S$-refutation $\rho$. The matrix is a bounded formula, so Parikh's theorem, in the form stated by Buss~\cite[Section~4.7, Theorem~11, p.~83]{Buss1986bounded}, gives a term $t(\pi)$ of bounded arithmetic such that $S^1_2$ proves the same sentence with the quantifier restricted to $\rho{\le}t(\pi)$.

The value of $t(\pi)$ need not be polynomial in $\pi$, but its binary length is polynomial in $|\pi|$. By induction on the structure of $t$, using $|x{+}y|{\le}\max(|x|,|y|){+}1$, $|x{\cdot}y|{\le}|x|{+}|y|$, $|x\#y|{=}|x||y|{+}1$ and the corresponding bounds for the remaining function symbols, there is a polynomial $q$ with nonnegative integer coefficients, which we may take with $q(n){\ge}n$, such that $S^1_2{\vdash}\forall\pi\,\forall n\,\forall\rho\,\bigl(|\pi|{\le}n\wedge\rho{\le}t(\pi)\rightarrow|\rho|{\le}q(n)\bigr)$. Reasoning in $S^1_2$, let $\pi$ be an $(\mathcal S{+}\phi)$-refutation with $|\pi|{\le}n$. The bounded witness $\rho{\le}t(\pi)$ then satisfies $Prf_{\mathcal S}(\rho,\bot)$ and $|\rho|{\le}q(n)$, so $\neg Con_{\mathcal S}(q(n))$. Hence $S^1_2{\vdash}\forall n\,\bigl(Con_{\mathcal S}(q(n)){\rightarrow}Con_{\mathcal S{+}\phi}(n)\bigr)$.

For the external consequence, fix this $S^1_2$-proof once and for all. Since $\mathcal S$ extends $S^1_2$, it yields an $\mathcal S$-proof of the same sentence. For a standard $n$, instantiate the universal quantifier at the binary numeral $\bar n$ and replace the term $q(\bar n)$ by the numeral $\overline{q(n)}$, using an $S^1_2$-proof of $q(\bar n){=}\overline{q(n)}$ of size polynomial in $n$. This gives $\mathcal S$-proofs of $Con_{\mathcal S}(q(n)){\rightarrow}Con_{\mathcal S{+}\phi}(n)$ of size polynomial in $n$. By hypothesis there are constants $c$ and $N$ such that $Con_{\mathcal S}(m)$ has an $\mathcal S$-proof of size at most $m^c$ for every $m{\ge}N$. For $n{\ge}N$ we have $q(n){\ge}N$, and modus ponens with that proof at $m{=}q(n)$ gives an $\mathcal S$-proof of $Con_{\mathcal S{+}\phi}(n)$ of size polynomial in $n$.
\end{proof}

Whether $EA{\vdash}Con_{\mathcal S}{\rightarrow}Con_{\mathcal S{+}\phi}$ suffices for simulation remains open here. Parikh's term-bounding argument does not apply to $EA$, which proves total exponentiation. The Friedman characterization~(Visser~\cite[Theorem~4.3]{Visser2017}) uses consistency restricted by formula complexity, not the ordinary consistency statements here, so it supplies no replacement argument. Theorem~\ref{thm:relative-consistency-implies-simulation} subsumes the verified interpretation mechanism, whose bounded implication yields the unbounded $S^1_2$ implication.

By contrast, Pudl\'ak~\cite[Problem~1, p.~193]{Pudlak1986length} conjectures that even the benchmark case $\phi{=}Con_{\mathcal S}$ already yields non-simulation; see also~Pudl\'ak~\cite[Section~3.1, p.~12]{Pudlak2017Author}. We use the following fixed-theory formulation:

\begin{conjecture}\label{conj:pudlak}
(Pudl\'ak's Conjecture for $\mathcal S$) Let $\mathcal S{\supseteq}S^1_2$ be sound with polynomial-time decidable axioms and its standard proof predicate. For every constant $c$, $\mathcal S{\centernot{\sststile{}{n^c}}}$ $Con_{\mathcal S{+}Con_{\mathcal S}}(n)$.
\end{conjecture}

Under Pudl\'ak's Conjecture for $\mathcal S$, any true extension that raises absolute consistency (proves $Con_{\mathcal S}$) is hard:

\begin{theorem}\label{thm:higher-absolute-consistency-hardness}
Assume Pudl\'ak's Conjecture for $\mathcal S$, and suppose $\mathcal S{+}\phi{\vdash}Con_{\mathcal S}$. Then for every constant $c$, $\mathcal S\centernot{\sststile{}{n^c}}Con_{\mathcal S{+}\phi}(n)$.
\end{theorem}

\begin{proof}
Replacing each use of $Con_{\mathcal S}$ by its fixed $(\mathcal S{+}\phi)$-proof gives an $S^1_2$-verified polynomial refutation translation. Hence, for a polynomial $p$, $\mathcal S$ has polynomial-size proofs of $Con_{\mathcal S{+}\phi}(p(n)){\rightarrow}Con_{\mathcal S{+}Con_{\mathcal S}}(n)$. Simulation of $\mathcal S{+}\phi$ would therefore give simulation of $\mathcal S{+}Con_{\mathcal S}$, contrary to the conjecture.
\end{proof}

In particular, the conjecture gives hardness for strictly stronger true extensions as well.

\begin{corollary}\label{cor:pudlak-not-tight}
Assume Pudl\'ak's Conjecture for a sound theory $\mathcal S{\supseteq}S^1_2$ with polynomial-time decidable axioms. Let $\theta$ be true but unprovable in $\mathcal S{+}Con_{\mathcal S}$, and put $\phi{:=}Con_{\mathcal S}{\land}\theta$. Then $\mathcal S{+}\phi$ is strictly stronger than $\mathcal S{+}Con_{\mathcal S}$, and $\mathcal S$ does not simulate $\mathcal S{+}\phi$.
\end{corollary}

\begin{proof}
Since $\mathcal S{+}\phi{\vdash}Con_{\mathcal S}$, non-simulation follows from Theorem~\ref{thm:higher-absolute-consistency-hardness}. The implication $\phi{\rightarrow}Con_{\mathcal S}$ is provable in $\mathcal S$, whereas $Con_{\mathcal S}{\rightarrow}\phi$ is not: such a proof would make $\mathcal S{+}Con_{\mathcal S}$ prove $\theta$, contrary to the hypothesis.
\end{proof}

Corollary~\ref{cor:pudlak-not-tight} shows upward closure of the conjectural hardness obstruction; it does not establish a converse characterization.

\subsection{Non-Simulation Implies Non-Simulation for Busy Beavers}\label{sec:busy-beaver-transfer}
A sound theory $\mathcal S{\supseteq}S^1_2$ with polynomial-time decidable axioms for which, for every true sentence $\phi$, there exists a polynomial bound on $\mathcal S$-proofs of $Con_{\mathcal S{+}\phi}(n)$ would simulate every true extension, even though $\mathcal S$ does not prove all true sentences. We show that any failure of this universal simulation is already witnessed by a canonical family of true sentences, namely exact Busy Beaver value statements adjoined to $S^1_2{+}\mathsf{Exp}$ (Theorem~\ref{thm:busy-beaver-transfer}).

For each $k$, let $t_k{=}BB(k)$ be the true $k$-state Busy Beaver value, and let $\phi_{BB}(k)$ denote the true sentence asserting the exact value $BB(k){=}t_k$. Here $BB(k)$ is the maximum halting time of any halting $k$-state machine on blank input, for a fixed deterministic machine convention with padding by unreachable states. The sentence $\phi_{BB}(k)$ implies that every halting $k$-state machine halts within $t_k$ steps; we use this bounded-halting consequence in the proof of Lemma~\ref{lem:busy-beaver-proves-consistency}.

We begin by proving that, for any fixed consistent theory with a polynomial-time ordinary proof predicate, sufficiently large Busy Beaver axioms imply its consistency over the base $S^1_2{+}\mathsf{Exp}$; the contradiction-search idea is taken from the proof of Aaronson's Proposition~3~\cite{AaronsonBusyBeaver}. The relevant threshold is the number of states of an explicitly compiled contradiction-search machine for the theory. Once this is in place, any hard true extension of a sound theory $\mathcal S{\supseteq}S^1_2$ with polynomial-time decidable axioms yields eventual hardness throughout the Busy Beaver family over $S^1_2{+}\mathsf{Exp}$.

\begin{definition}\label{def:contradiction-search}
(\textbf{Contradiction Search Threshold}) Let $\mathcal T$ have a polynomial-time ordinary proof predicate, and let $E_{\mathcal T}$ be the fixed Turing machine that, on blank input, searches for a $\mathcal T$-proof of contradiction. Define $k_{\mathcal T}^{\mathrm{cs}}$ to be the number of states of $E_{\mathcal T}$. For each $k{\ge}k_{\mathcal T}^{\mathrm{cs}}$, let $E_{\mathcal T,k}$ be the $k$-state machine obtained by adding unreachable states.
\end{definition}

\begin{lemma}\label{lem:busy-beaver-proves-consistency}
Let $\mathcal T$ be a consistent theory with a polynomial-time ordinary proof predicate. Then for every standard $k{\ge}k_{\mathcal T}^{\mathrm{cs}}$, one has $S^1_2{+}\mathsf{Exp}{+}\phi_{BB}(k){\vdash}Con_{\mathcal T}$.
\end{lemma}

\begin{proof}
Implement $E_{\mathcal T}$ by enumerating candidate $\mathcal T$-proofs in length-lexicographic order, running the clocked verifier for $Prf_{\mathcal T}(\cdot,\ulcorner0{=}1\urcorner)$ on each candidate, and halting exactly when it accepts.

Use a canonical binary tableau encoding: a history of $\ell$ steps contains exactly $\ell{+}1$ configurations, recording the state, head position and the tape segment $[-\ell,\ell]$, with self-delimiting headers and no additional padding. Write $\ell(h)$ for the step count and $\mathrm{HaltTrace}_M(h)$ for the polynomial-time check that this is a complete halting history of $M$ on blank input. Its bounded-arithmetic representation includes the format check. Put $\mathrm{Halt}_M\equiv\exists h\,\mathrm{HaltTrace}_M(h)$ and $\mathrm{HaltBy}_M(t)\equiv\exists h\,(\mathrm{HaltTrace}_M(h)\land\ell(h){\le}t)$. The format gives a fixed polynomial $q_M$ such that $S^1_2$ proves $\mathrm{HaltTrace}_M(h)\land\ell(h){\le}t \ \longrightarrow\ |h|{\le}q_M(t)$.

For each fixed standard $t$, choose the standard integer $H_{M,t}{=}2^{q_M(t)+1}$. Then $S^1_2$ proves the fixed-parameter equivalence $\mathrm{HaltBy}_M(\bar t)\ \leftrightarrow\ \allowbreak\exists h{<}\overline{H_{M,t}}\, (\mathrm{HaltTrace}_M(h)\land\allowbreak\ell(h){\le}\bar t)$. Here $\overline{H_{M,t}}$ is a fixed numeral; no bounded-arithmetic term for $2^{q_M(t)+1}$ with variable $t$, and no uniformly bounded formula for $\mathrm{HaltBy}_M(t)$, is asserted.

To formalize $\phi_{BB}(k)$, let $\mathrm{HaltTrace}(M,h)$ be the corresponding tableau formula uniform in the machine code, so that $\mathrm{HaltTrace}_M(h)$ abbreviates $\mathrm{HaltTrace}(\ulcorner M\urcorner,h)$. Writing $\mathrm{Mach}_k$ for the set of codes of $k$-state machines, $\phi_{BB}(k)$ is the conjunction of $\forall M{\in}\mathrm{Mach}_k\,\forall h\,(\mathrm{HaltTrace}(M,h){\rightarrow}\ell(h){\le}t_k)$ and $\exists M{\in}\mathrm{Mach}_k\,\exists h\,(\mathrm{HaltTrace}(M,h){\land}\ell(h){=}t_k)$. The first clause bounds every halting computation; the second asserts attainment of the bound. The numeral $t_k$ is fixed before the proof-length parameter $n$, and no uniformly feasible computation of $t_k$ from $k$ is asserted.

First, $S^1_2{+}\mathsf{Exp}$ proves that every $\mathcal T$-refutation yields a halting history of $E_{\mathcal T}$. Let $q$ be a polynomial bounding the running time of the clocked verifier. If $p$ is a $\mathcal T$-refutation, then $E_{\mathcal T}$ examines fewer than $2^{|p|+1}$ candidates before it reaches $p$ and spends at most $q(|p|)$ steps on each, up to a fixed coding factor, so it halts within $2^{|p|+1}q(|p|)$ steps up to that factor; the complete history up to halting has size bounded by a fixed polynomial in this quantity. Choose a fixed polynomial $r$ such that $2^{r(|p|)}$ exceeds this size bound. There is a fixed polynomial-time function $F_{\mathcal T}$ which, given $p$ and a padding string $z$, writes out the history of $E_{\mathcal T}$ up to halting, and $S^1_2$ proves
\[
Prf_{\mathcal T}(p,\ulcorner0{=}1\urcorner)\wedge E_2(r(|p|),a)\wedge E_2(a,z)\ \rightarrow\ \mathrm{HaltTrace}_{E_{\mathcal T}}\bigl(F_{\mathcal T}(p,z)\bigr).
\]
Indeed, $E_2(a,z)$ gives $|z|{=}a{+}1$. Let $F_{\mathcal T}(p,z)$ run $E_{\mathcal T}$ on blank input for at most $|z|$ steps, recording the successive configurations, and output a fixed dummy string if $E_{\mathcal T}$ has not halted by then; its running time and output length are polynomial in the length of its padded input. $S^1_2$ proves by $\Sigma^b_1$ length induction~(Buss~\cite[Chapter~2]{Buss1986bounded}) on the step counter $i{\le}|z|$ that the first $i$ recorded configurations form a valid computation of $E_{\mathcal T}$, and by the same induction on the index $j$ of the candidate under examination, which stays below $2^{|p|+1}{\le}a{<}|z|$, that the search has either halted at an earlier candidate or reaches $p$ within the step bound above. Since $Prf_{\mathcal T}(p,\ulcorner0{=}1\urcorner)$ holds and the clocked verifier is $S^1_2$-provably equivalent to $Prf_{\mathcal T}$, the verifier accepts at $p$ at the latest. Hence the clock does not expire, and the output is a complete halting history. Already in $S^1_2$ the value $a{=}2^{r(|p|)}$ exists: its binary length is polynomial in $|p|$, and it is obtained by a polynomial-time function, with a term bound supplied by smash. The axiom $\mathsf{Exp}$ is used to obtain $z{=}2^a$, whose length is exponential in $|p|$. Hence $S^1_2{+}\mathsf{Exp}\vdash\forall p\,\bigl(Prf_{\mathcal T}(p,\ulcorner0{=}1\urcorner)\rightarrow\mathrm{Halt}_{E_{\mathcal T}}\bigr)$.

The induction invariant asserts either that a valid halting prefix has already been produced, or that the search prefix is valid, with its current candidate and verifier configuration, and every completed verifier run has rejected. Prefixes and verifier tableaux have size polynomial in $|p|{+}|z|$ and therefore have bounded-arithmetic term bounds supplied by smash. Their polynomial-time format and transition checks have $\Sigma^b_1$ representations, so the invariant may be taken in that class. All loop indices are bounded by the length of the padding parameter $z$; the argument uses length induction with that parameter, rather than unrestricted induction through an exponential numerical bound.

Now fix a standard $k{\ge}k_{\mathcal T}^{\mathrm{cs}}$ and let $U{:=}S^1_2{+}\mathsf{Exp}{+}\phi_{BB}(k)$. Three further facts hold:
\[
\begin{aligned}
S^1_2&\vdash \mathrm{Halt}_{E_{\mathcal T}}\rightarrow\mathrm{Halt}_{E_{\mathcal T,k}},\\
S^1_2{+}\phi_{BB}(k)&\vdash \mathrm{Halt}_{E_{\mathcal T,k}}\rightarrow\mathrm{HaltBy}_{E_{\mathcal T,k}}(t_k),\\
S^1_2&\vdash \neg\mathrm{HaltBy}_{E_{\mathcal T,k}}(t_k).
\end{aligned}
\]
The first follows because padding adds only unreachable states, and $S^1_2$ verifies the corresponding map of computation histories. The second follows by instantiating the upper-bound clause of $\phi_{BB}(k)$ at $E_{\mathcal T,k}$: $S^1_2$ proves the true bounded fact $\ulcorner E_{\mathcal T,k}\urcorner{\in}\mathrm{Mach}_k$. For the third, consistency of $\mathcal T$ implies that $E_{\mathcal T,k}$ does not halt within its first $t_k$ steps; the fixed-parameter equivalence established above makes $\neg\mathrm{HaltBy}_{E_{\mathcal T,k}}(\overline{t_k})$ equivalent in $S^1_2$ to a true closed bounded sentence. Exhaustive verification of that bounded sentence gives an $S^1_2$-proof, possibly of enormous length. This uses the fixed standard parameters $k,t_k$, not a claim that every true $\Pi_1$ sentence is provable. Reasoning in $U$, a $\mathcal T$-refutation would yield a halting history of $E_{\mathcal T}$, hence of $E_{\mathcal T,k}$, hence a halting history with at most $t_k$ steps, contradicting the third fact. Therefore $U{\vdash}Con_{\mathcal T}$.
\end{proof}

\paragraph{Choice of base.} We use $S^1_2{+}\mathsf{Exp}$ as the base of the Busy Beaver witness theories; the simulating theory need only extend $S^1_2$. Whether Lemma~\ref{lem:busy-beaver-proves-consistency} holds over $S^1_2$ alone remains open. The base $EA$ also suffices under the usual definitional translation.

If a consistent $\mathcal T{\supseteq}S^1_2$ with a polynomial-time ordinary proof predicate also proves $\mathsf{Exp}$, the second incompleteness theorem therefore makes $\phi_{BB}(k)$ unprovable in $\mathcal T$ for every $k{\ge}k_{\mathcal T}^{\mathrm{cs}}$.

With this result, we can show that the existence of some hard true extension $\mathcal S{+}\phi$ is equivalent to eventual hardness throughout the Busy Beaver family over $S^1_2{+}\mathsf{Exp}$:

\begin{theorem}\label{thm:busy-beaver-transfer}
Let $\mathcal S{\supseteq}S^1_2$ be a sound theory with polynomial-time decidable axioms. The following are equivalent:
\begin{enumerate}
\item There exists a true sentence $\phi$ such that for every constant $c$, $\mathcal S\centernot{\sststile{}{n^c}}Con_{\mathcal S{+}\phi}(n)$.
\item For all sufficiently large $k$, $\mathcal S$ does not simulate $S^1_2{+}\mathsf{Exp}{+}\phi_{BB}(k)$.
\end{enumerate}
\end{theorem}

\begin{proof}
$(1){\rightarrow}(2)$ Put $\mathcal T{:=}\mathcal S{+}\phi$. It is consistent and has a polynomial-time ordinary proof predicate. For each fixed $k{\ge}k_{\mathcal T}^{\mathrm{cs}}$, Lemma~\ref{lem:busy-beaver-proves-consistency} gives a fixed proof $\pi_k$ of $Con_{\mathcal T}$ in $U_k{:=}S^1_2{+}\mathsf{Exp}{+}\phi_{BB}(k)$.

Given a $\mathcal T$-refutation $\rho$, the verifier-to-proof construction of Section~\ref{sec:preliminaries} supplies a $U_k$-proof of $Prf_{\mathcal T}(\bar\rho,\ulcorner0{=}1\urcorner)$ of size polynomial in $|\rho|$. Appending $\pi_k$ yields a $U_k$-refutation. The construction is verified in $S^1_2$, so for a polynomial $p_k(n){\ge}n$, $\mathcal S$ has polynomial-size proofs of $Con_{U_k}(p_k(n)){\rightarrow}Con_{\mathcal T}(n)$. Simulation of $U_k$ would then imply simulation of $\mathcal T$, contradicting (1). Exponentiation enters through the fixed proof $\pi_k$; it is not required in $\mathcal S$.

$(2){\rightarrow}(1)$ For any $k$ supplied by (2), take the true sentence $\phi{:=}\mathsf{Exp}{\land}\phi_{BB}(k)$. The specified embedding of $S^1_2$ into $\mathcal S$ extends to an $S^1_2$-verified polynomial embedding of $U_k$-proofs into $(\mathcal S{+}\phi)$-proofs. Thus for a polynomial $p(n){\ge}n$, $\mathcal S$ has polynomial-size proofs of $Con_{\mathcal S{+}\phi}(p(n)){\rightarrow}Con_{U_k}(n)$. Simulation of $\mathcal S{+}\phi$ would contradict (2).
\end{proof}

In Theorem~\ref{thm:busy-beaver-transfer}, hardness holds for every sufficiently large $k$ and, for each such $k$ and each exponent, infinitely many consistency parameters $n$.

We next ask whether the asymptotic Busy Beaver transfer also holds when hardness is assumed at just one finite parameter. The answer is yes: for a fixed proving theory, an exponential lower bound at a single sufficiently large $N$ transfers to a Busy Beaver consistency statement at a polynomially larger parameter, with a lower bound still exponential in the original $N$.

\begin{theorem}[Finite Busy Beaver transfer]\label{thm:finite-busy-beaver-transfer}
Let $\mathcal S{\supseteq}S^1_2$ be sound with a polynomial-time ordinary proof predicate, let $\phi$ be true, and put $\mathcal T{:=}\mathcal S{+}\phi$. Write $\ell_{\mathcal S}(\sigma)$ for the length of a shortest $\mathcal S$-proof of $\sigma$, with value $\infty$ if no proof exists. Fix $k{\ge}k_{\mathcal T}^{\mathrm{cs}}$ and put $U_k{:=}S^1_2{+}\mathsf{Exp}{+}\phi_{BB}(k)$. There is an integer polynomial $p_k(N){\ge}N$ such that, for every $a{>}0$, there are $b{>}0$ and $N_0$ for which every $N{\ge}N_0$ satisfies
\[
\ell_{\mathcal S}\bigl(Con_{\mathcal T}(N)\bigr){>}2^{aN}
\quad\Longrightarrow\quad
\ell_{\mathcal S}\bigl(Con_{U_k}(p_k(N))\bigr){>}2^{bN}.
\]
If $\mathcal S{\vdash}\mathsf{Exp}$, the same assertion holds with $U_k$ replaced by $\mathcal S{+}\phi_{BB}(k)$, with possibly different parameters.
\end{theorem}

\begin{proof}
Fix the $U_k$-proof $\pi_k$ of $Con_{\mathcal T}$ supplied by Lemma~\ref{lem:busy-beaver-proves-consistency}. As in the proof of Theorem~\ref{thm:busy-beaver-transfer}, the verifier-to-proof transformation, with $\pi_k$ appended, gives an $S^1_2$-verified polynomial refutation translation from $\mathcal T$ to $U_k$. Hence for a polynomial $p_k$ there are $\mathcal S$-proofs of $Con_{U_k}(p_k(N)){\rightarrow}Con_{\mathcal T}(N)$ of size polynomial in $N$. Enlarge the bound to cover every standard $N$, including the finitely many small values. Composing such a proof with a proof of its antecedent of length at most $L$ has length bounded by an integer polynomial $r_k(N,L)$, chosen nondecreasing in its nonnegative arguments. Contraposition gives the general finite bound: for every $N,L{\in}\mathbb N$, if $\ell_{\mathcal S}(Con_{\mathcal T}(N)){>}r_k(N,L)$, then $\ell_{\mathcal S}(Con_{U_k}(p_k(N))){>}L$. For an append-only proof encoding, concatenating the implication proof and one application of modus ponens gives $r_k(N,L){=}L{+}q_k(N)$ for a polynomial $q_k$; otherwise the polynomial bound suffices.

Choose constants $C,d{\ge}1$ such that $r_k(N,L){\le}C(N{+}L{+}1)^d$. For any $a{>}0$, choose $0{<}b{<}a/d$. Then $r_k(N,\lfloor2^{bN}\rfloor){\le}2^{aN}$ for all sufficiently large $N$. Apply the finite bound with $L{=}\lfloor2^{bN}\rfloor$; since proof lengths are integers, the conclusion is the asserted exponential lower bound.

For the same-base assertion, substitute fixed $\mathcal S$-proofs of the base axioms, including $\mathsf{Exp}$, in $\pi_k$. This gives an $(\mathcal S{+}\phi_{BB}(k))$-proof of $Con_{\mathcal T}$, to which the same argument applies. Without the hypothesis $\mathcal S{\vdash}\mathsf{Exp}$, use $\mathcal S{+}(\mathsf{Exp}{\land}\phi_{BB}(k))$ instead.
\end{proof}

The theorem applies to a single $N{\ge}N_0$ at which the stated lower bound holds; it assumes no lower bounds at other parameters. Its conclusion is exponential in the original $N$. If $p_k$ has degree $D$, the guarantee in the new parameter $M{=}p_k(N)$ is in general only $2^{\Omega(M^{1/D})}$. Thus polynomial reindexing need not preserve a fully exponential bound in $M$.

Conversely, fix any standard $k$ and put $\phi_k{:=}\mathsf{Exp}{\land}\phi_{BB}(k)$. The feasible embedding of $U_k$ into $\mathcal S{+}\phi_k$ gives nondecreasing integer polynomials $p'_k(N){\ge}N$ and $r'_k(N,L)$ such that $\ell_{\mathcal S}(Con_{\mathcal S{+}\phi_k}(p'_k(N))){\le}L$ implies $\ell_{\mathcal S}(Con_{U_k}(N)){\le}r'_k(N,L)$: translate refutations, instantiate the resulting bounded-consistency implication, and compose proofs. Contraposition transfers a finite Busy Beaver lower bound to the true extension $\mathcal S{+}\phi_k$, which need not be the original $\mathcal S{+}\phi$.

The forward constants and $N_0$ depend on the fixed $k$ and its consistency proof $\pi_k$. Thus the transferred bound depends on the chosen Busy Beaver instance, and a single finite lower bound gives no asymptotic hardness conclusion.

Under Pudl\'ak's Conjecture, one can give a stronger statement indicating just how large $k$ must be:

\begin{theorem}\label{thm:pudlak-implies-busy-beaver-hardness}
Assume Pudl\'ak's Conjecture for $\mathcal S$, and suppose that $\mathcal S{\vdash}\mathsf{Exp}$. Then for every $k{\ge}k_{\mathcal S}^{\mathrm{cs}}$, for every constant $c$, $\mathcal S\centernot{\sststile{}{n^c}}Con_{\mathcal S{+}\phi_{BB}(k)}(n)$.
\end{theorem}

\begin{proof}
Fix $k{\ge}k_{\mathcal S}^{\mathrm{cs}}$. Lemma~\ref{lem:busy-beaver-proves-consistency} gives $S^1_2{+}\mathsf{Exp}{+}\phi_{BB}(k){\vdash}Con_{\mathcal S}$. Since $\mathcal S$ extends $S^1_2$ and proves $\mathsf{Exp}$, fixed proof substitution gives $\mathcal S{+}\phi_{BB}(k){\vdash}Con_{\mathcal S}$. Theorem~\ref{thm:higher-absolute-consistency-hardness}, using the assumed instance of Pudl\'ak's conjecture for $\mathcal S$, gives the result.
\end{proof}

Without assuming $\mathcal S{\vdash}\mathsf{Exp}$, Pudl\'ak's conjecture still gives non-simulation of $S^1_2{+}\mathsf{Exp}{+}\phi_{BB}(k)$ for every $k{\ge}k_{\mathcal S{+}Con_{\mathcal S}}^{\mathrm{cs}}$, by applying Theorem~\ref{thm:busy-beaver-transfer} to $\phi{=}Con_{\mathcal S}$. These are sufficient thresholds; the results do not classify the smaller Busy Beaver instances.

Combining Busy Beaver transfer with the Kraj\'i\v{c}ek--Pudl\'ak correspondence yields the paper's main structural characterization of the nonexistence of a length-optimal propositional proof system.

\begin{corollary}\label{cor:optimal-characterization}
The following are equivalent:
\begin{enumerate}
\item No optimal propositional proof system exists.
\item For every sound theory $\mathcal S{\supseteq}S^1_2$ with polynomial-time decidable axioms, there is a true sentence $\phi$ such that for every constant $c$, $\mathcal S\centernot{\sststile{}{n^c}}Con_{\mathcal S{+}\phi}(n)$.
\item For every sound theory $\mathcal S{\supseteq}S^1_2$ with polynomial-time decidable axioms, there is $K_{\mathcal S}$ such that for every $k{\ge}K_{\mathcal S}$ and every constant $c$, $\mathcal S\centernot{\sststile{}{n^c}}Con_{S^1_2{+}\mathsf{Exp}{+}\phi_{BB}(k)}(n)$.
\end{enumerate}
The equivalence is unchanged if the theories $\mathcal S$ quantified in (2) and (3) are restricted to sound finitely axiomatized sequential extensions of $S^1_2$.
\end{corollary}

\begin{proof}
Conditions (2) and (3) are equivalent by Theorem~\ref{thm:busy-beaver-transfer}. For $(1){\rightarrow}(2)$, the Kraj\'i\v{c}ek--Pudl\'ak correspondence~\cite[Theorem~2.1]{Krajicek}, in the formulation of Khaniki~\cite[Theorem~2.3]{khaniki2019}, gives a consistent $\mathcal T{\supseteq}S^1_2$ with polynomial-time decidable axioms whose bounded consistency has no polynomial-size $\mathcal S$-proofs. Take $\phi{:=}Con_{\mathcal T}$. Verifying a $\mathcal T$-refutation and applying this axiom gives an $(\mathcal S{+}\phi)$-refutation with polynomial overhead, so simulation of $\mathcal S{+}\phi$ would imply simulation of $\mathcal T$.

Conversely, an optimal proof system $P$ gives the sound theory $\mathcal S_P{:=}S^1_2{+}\mathrm{RFN}(P)$, which simulates every consistent theory in the class by the classical reflection construction~(Kraj\'i\v{c}ek and Pudl\'ak~\cite[construction~(a2) and proof of Theorem~2.1]{Krajicek}). Specifically, for each fixed clocked verifier of a target theory $\mathcal T$, its Cook tableau encodings of $Con_{\mathcal T}(n)$ have polynomial size. The encoding construction gives polynomial-size $S^1_2$-proofs that their truth under every assignment implies $Con_{\mathcal T}(n)$. Optimality gives polynomial-size $P$-proofs of these tautologies when $\mathcal T$ is consistent. Verifying each such proof and using $\mathrm{RFN}(P)$ gives the required $\mathcal S_P$-proofs. In particular $\mathcal S_P$ simulates all its true single-sentence extensions, contradicting (2). Since $\mathcal S_P$ is finitely axiomatized and sequential, the same argument proves the assertion for the restricted class.
\end{proof}

The propositional consequence applies to every Cook--Reckhow proof system $P$, under the equivalent hypotheses of Corollary~\ref{cor:optimal-characterization}. Let $\tau_{k,n}$ be the standard Cook encoding of $Con_{U_k}(n)$, where $U_k{:=}S^1_2{+}\mathsf{Exp}{+}\phi_{BB}(k)$. For each fixed $k$ this is a family of tautologies of size polynomial in $n$. Take the sound finite theory $\mathcal S_P{:=}S^1_2{+}\mathrm{RFN}(P)$. A $P$-proof of $\tau_{k,n}$ gives an $\mathcal S_P$-proof of $Con_{U_k}(n)$ with polynomial overhead, by verifying the proof and applying $\mathrm{RFN}(P)$. Condition~(3) therefore implies that, for every sufficiently large $k$, the family $(\tau_{k,n})_n$ has no polynomial bound on its $P$-proof lengths. The collection uses the true exact Busy Beaver values; neither these values nor the threshold in $k$ are supplied effectively. This consequence uses the global no-optimality hypothesis, whereas Theorem~\ref{thm:finite-busy-beaver-transfer} transfers individual bounded-consistency lower bounds in a fixed arithmetic theory.

A diagonal argument using the Busy Beaver transfer theorem yields a single Busy Beaver family hard for every theory in a specified class:

\begin{corollary}\label{cor:diagonal-busy-beaver-family}
Let $\mathfrak C$ be any nonempty class of sound theories $\mathcal S{\supseteq}S^1_2$ with polynomial-time decidable axioms, each having a true extension not simulated by $\mathcal S$. Then there is a nonconstructively chosen function $k{:}\mathbb N{\to}\mathbb N$ with unbounded range such that, for every $\mathcal S{\in}\mathfrak C$ and every constant $c{>}0$, $\mathcal S\centernot{\sststile{}{n^c}}Con_{U_{k(n)}}(n)$, where $U_j{:=}S^1_2{+}\mathsf{Exp}{+}\phi_{BB}(j)$. The hypothesis concerns only the theories in $\mathfrak C$; taking all eligible theories gives the universal hypothesis of Corollary~\ref{cor:optimal-characterization}.
\end{corollary}

\begin{proof}
List all pairs $(\mathcal S_e,d_e)$ with $\mathcal S_e{\in}\mathfrak C$ and positive integer $d_e$, repeating each pair infinitely often. The list need not be effective. By Theorem~\ref{thm:busy-beaver-transfer}, choose strictly increasing $k_e$ such that the fixed family $(Con_{U_{k_e}}(n))_n$ has no polynomial bound on its $\mathcal S_e$-proof lengths. Then choose strictly increasing $n_e$, with $n_e$ at least the written length of $\phi_{BB}(k_e)$, such that $\ell_{\mathcal S_e}(Con_{U_{k_e}}(n_e)){>}e n_e^{d_e}$. Put $k(n_e){:=}k_e$ and $k(n){:=}1$ elsewhere. An eventual proof bound $A n^c$ for a theory $\mathcal S{\in}\mathfrak C$ is contradicted at a repeated pair $(\mathcal S,d)$ with integer $d{\ge}c$, stage $e{>}A$, and $n_e$ beyond the proposed onset. The strictly increasing $k_e$ give unbounded range.
\end{proof}

This conclusion applies to classes of theories with hard true extensions even when other theories are universal simulators. It does not identify these classes with all non-optimal Cook--Reckhow proof systems: adjoining $\mathrm{RFN}(P)$ supplies a theory that proves $P$ sound, but does not in general preserve $P$'s proof lengths in the reverse direction~(Pudl\'ak~\cite[Section~4.1]{Pudlak2020reflectionprinciples}).

\subsection{Alternative Canonical Witnesses}\label{sec:alternative-witnesses}
Theorem~\ref{thm:busy-beaver-transfer} can also be stated with exact Busy Beaver values replaced by exact prefixes of Chaitin's halting probability or complete finite lists of Kolmogorov-random strings, using $\mathsf{PA}$ as the witness base. Fix optimal plain and prefix-free machines $U_0,V$ with compiler simulations verified in $\mathsf{PA}$. Let $\chi_\Omega(r)$ assert the first $r$ bits of $\Omega_V{:=}\sum_{V(p)\downarrow}2^{-|p|}$; they determine halting for programs of length at most $r$~(Calude, Dinneen and Shu~\cite[Theorem~4]{CaludeDinneenShu2002Omega}). For a fixed positive rational $d$, put $g_d(m){:=}\max\{0,\lceil m{-}d\log_2m\rceil\}$. Let $\chi_R(r)$ conjoin the true statements $K_{U_0}(x){\ge}g_d(|x|)$ or their negations for every $x$ with $2{\le}|x|{\le}r$, specifying a complete correct list. For related uses of finite complexity information, see Bienvenu et al.~\cite[Section~3.1 of the preprint]{Bienvenu2014}.

\begin{proposition}[Alternative canonical witnesses]\label{prop:alternative-canonical-witnesses}
Let $\mathcal S{\supseteq}S^1_2$ be sound with polynomial-time decidable axioms. Some true extension of $\mathcal S$ is not simulated if and only if $\mathcal S$ fails to simulate $\mathsf{PA}{+}\chi_\Omega(r)$ for every sufficiently large $r$, and equivalently if and only if it fails to simulate $\mathsf{PA}{+}\chi_R(r)$ for every sufficiently large $r$.
\end{proposition}
\begin{proof}
It suffices to show that, for each consistent theory $\mathcal T$ with a polynomial-time ordinary proof predicate, either witness theory $\mathcal V_r$ proves $Con_{\mathcal T}$ for all sufficiently large $r$. The proof transformation in Theorem~\ref{thm:busy-beaver-transfer} then gives the forward implications. Conversely, if $\mathcal V_r$ is not simulated, take the true sentence $\phi{:=}Con_{\mathcal V_r}$. Verifying a $\mathcal V_r$-refutation and applying this axiom gives an $(\mathcal S{+}\phi)$-refutation of polynomially larger size, so simulation of $\mathcal S{+}\phi$ would imply simulation of $\mathcal V_r$.

For the $\Omega$ family, put $q_r{:=}2^{-r}\lfloor2^r\Omega_V\rfloor$ and let $\omega_s$ be the increasing rational approximation obtained by enumerating halting $V$-programs. Using coded rational comparisons, formalize $\chi_\Omega(r)$ as $(\exists s\ \omega_s{>}q_r)\land(\exists j\ \forall s\ \omega_s{\le}q_r{+}2^{-r}{-}2^{-j})$, with $q_r$ fixed data. Compile $E_{\mathcal T}$ into a $V$-program $p_{\mathcal T}$ with the simulation verified in $\mathsf{PA}$. For $r{\ge}|p_{\mathcal T}|$, choose a standard stage $s$ with $\omega_s{>}q_r$. Consistency of $\mathcal T$ ensures that $p_{\mathcal T}$ is absent at that stage; both finite facts are verifiable in $\mathsf{PA}$. Its later halting would add at least $2^{-r}$ to the sum, contradicting $\chi_\Omega(r)$. Thus $\mathsf{PA}{+}\chi_\Omega(r){\vdash}Con_{\mathcal T}$.

For the list family, let $D_t(m)$ contain the length-$m$ outputs produced within $t$ steps by $U_0$-programs shorter than $g_d(m)$. Fix a standard stage $s_m$ by which all genuinely non-random strings of length $m$ have appeared. The axiom $\chi_R(m)$ and these finite computations prove that, for $t{\ge}s_m$, $D_t(m)$ is exactly the set of non-random strings of length $m$. Given $m$, run $E_{\mathcal T}$; if it halts after $t$ steps, output the first length-$m$ string outside $D_t(m)$. Such a string exists because fewer than $2^{g_d(m)}{\le}2^m$ programs are considered. This algorithm has a compiled description of length at most $2\lceil\log_2m\rceil{+}c_{\mathcal T}{<}g_d(m)$ for all sufficiently large $m$. Halting at $t{\ge}s_m$ would therefore produce a short description of a string asserted random by $\chi_R(m)$, a contradiction. Consistency also gives a $\mathsf{PA}$-proof that $E_{\mathcal T}$ does not halt before the fixed stage $s_m$. Formalizing these arguments in $\mathsf{PA}$ gives $\mathsf{PA}{+}\chi_R(m){\vdash}Con_{\mathcal T}$.
\end{proof}

Either family can therefore replace Busy Beaver values in Corollary~\ref{cor:optimal-characterization}, with witness base $\mathsf{PA}$. The list argument also works for fixed deficiency. Finite transfers have bounds depending on the fixed witness and its consistency proof; no feasible bound uniform in $r$ or reduction of the witness base is established. Hardness of a complete list does not imply hardness of its individual axioms.

\section{Refutation of the Higher Relative Consistency Conjecture}\label{sec:desiderata}
In the first version of this paper, we conjectured under the name Higher Relative Consistency (HRC) that, for sound finitely axiomatized sequential $\mathcal S{\supseteq}S^1_2$, simulation of a true extension $\mathcal S{+}\phi$ requires $EA{\vdash}Con_{\mathcal S}{\rightarrow}Con_{\mathcal S{+}\phi}$. Hrube\v{s}'s self-reference construction refutes this necessary condition. We use his result as presented in Pudl\'ak's survey~\cite[Theorem~3.4, p.~13]{Pudlak2017Author}.

\begin{theorem}[Hrube\v{s}]\label{thm:hrubes}
Let $\mathcal S$ be a consistent, finitely axiomatized sequential extension of $S^1_2$, and let $\mathcal T{\supseteq}\mathcal S$ be consistent with polynomial-time decidable axioms and standard proof predicates. There is a true sentence $\rho$, equivalent over $EA$ to a $\Pi_1$ sentence, such that $\mathcal T{\not\vdash}\rho$ while $\mathcal S$ has polynomial-size proofs of $Con_{\mathcal S{+}\rho}(n)$.
\end{theorem}

To refute HRC, choose a sound finitely axiomatized sequential $\mathcal S$ containing $EA$ and $S^1_2$ in the common definitional language of Section~\ref{sec:preliminaries}. For example, use a finite axiomatization of $I\Sigma_1$ with the chosen finite $S^1_2$ axioms adjoined~(H\'ajek and Pudl\'ak~\cite[Chapter~I, Theorem~2.52]{HajekPudlak}). Put $\mathcal T{:=}\mathcal S{+}Con_{\mathcal S}$, which is sound and has polynomial-time decidable axioms. Theorem~\ref{thm:hrubes} gives a true sentence $\rho$ such that $\mathcal S$ simulates $\mathcal S{+}\rho$ but $\mathcal T{\not\vdash}\rho$. Since $\mathcal S$ contains $EA$, we may replace $\rho$ by its $\Pi_1$ form, preserving these properties by fixed proof substitution.

Suppose $EA{\vdash}Con_{\mathcal S}{\rightarrow}Con_{\mathcal S{+}\rho}$. Then $\mathcal T{\vdash}Con_{\mathcal S{+}\rho}$. Formalized $\Sigma_1$-completeness gives $\mathcal S{\vdash}\neg\rho{\rightarrow}Prov_{\mathcal S}(\ulcorner\neg\rho\urcorner)$, hence $\mathcal S{\vdash}Con_{\mathcal S{+}\rho}{\rightarrow}\rho$~(Artemov and Beklemishev~\cite[Proposition~16]{ArtemovBeklemishev2005}). Thus $\mathcal T{\vdash}\rho$, a contradiction. We have a simulated true extension without the $EA$ relative-consistency certificate required by HRC.

The same example has no such certificate in $S^1_2$, since every $S^1_2$-proof yields an $EA$-proof under our language conventions. The same contradiction follows if the implication is provable in $\mathcal S$ itself, so Hrube\v{s}'s construction also rules out the characterization that $\mathcal S$ simulates $\mathcal S{+}\phi$ if and only if $\mathcal S{\vdash}Con_{\mathcal S}{\rightarrow}Con_{\mathcal S{+}\phi}$. The sufficient condition in Theorem~\ref{thm:relative-consistency-implies-simulation} is unaffected: $S^1_2$-provable relative consistency yields simulation when $\mathcal S$ has short self-consistency proofs.

\section{Kolmogorov Blindness}\label{sec:higher-relative-consistency}
A natural refinement of the Busy Beaver phenomenon is obtained by considering Kolmogorov-random strings. Instead of an axiom specifying an exact Busy Beaver value, we consider axioms asserting that particular strings have no substantially shorter program descriptions. The set of exact Busy Beaver pairs $\{(k,BB(k)){:}k{\ge}1\}$ and the randomness set $R^{\log}$ defined below are both \emph{immune}: each is infinite but contains no infinite computably enumerable subset.\footnote{An infinite enumeration of exact Busy Beaver pairs would give a computable upper bound on $BB(n)$ by finding a pair with first coordinate at least $n$, contradicting Aaronson~\cite[Proposition~1]{AaronsonBusyBeaver}. An infinite enumeration of $R^{\log}$ would allow a program given $L$ to output a member of length at least $L$, with description length $O(\log L)$, contradicting its logarithmic-deficiency randomness for large $L$; this is the compression argument underlying Chaitin's incompleteness theorem (see Shen~\cite[Theorem~4]{Shen2015Kolmogorov}).} Consequently, a fixed sound effectively axiomatized arithmetic theory can prove membership for only finitely many elements of either set. This supports the intuition that a theory has limited information about $R^{\log}$ and motivates a hardness conjecture. Immunity alone gives no proof-length lower bound.

\subsection{Randomness and the Blindness Conjecture}

Kolmogorov Blindness predicts hardness at individual finite scales, paralleling finite Busy Beaver transfer: exponential bounds beyond an onset common to true randomness axioms of each sufficiently large string length, and hence dense hardness because $R^{\log}$ has density tending to one. The bounds are also conjectured to hold after adjoining a true randomness axiom of the same length when the target remains random conditional on the added string. This is stronger than dense hardness obtained by padding $\phi_{BB}(k)$ with $|x|{>}0$: all those padded axioms are equivalent, so adjoining one gives short bounded-consistency proofs for the others over a finitely axiomatized sequential base.

The motivation comes from random DNF hardness candidates~(Khaniki~\cite[Introduction]{Khaniki2024JumpOperators}) and proof-complexity generators~(Alekhnovich et al.~\cite{Alekhnovichetal}; Kraj\'i\v{c}ek~\cite{Krajicek2024Generators}). Here incompressibility specifies the indices and the bounds concern finite consistency parameters. The conditional-randomness requirement is a further conjectural strengthening proposed here.

Fix an optimal machine $U$ and a positive rational constant $d$. For each string $x{\in}\{0,1\}^*$, let $K_U(x)$ denote the plain Kolmogorov complexity of $x$, that is, the length of the shortest program $p$ such that $U(p){=}x$. Logs are to base two. Put $g(m){:=}\max\{0,\lceil m{-}d\log_2 m\rceil\}$ for $m{\ge}2$ and $R^{\log}{:=}\{x{:}|x|{\ge}2,\ K_U(x){\ge}g(|x|)\}$. For rational $d$, the integer threshold $g$ is computable by integer comparisons; no real-number parameter enters the arithmetic sentence. The rounding affects only constant factors in counting estimates. Fewer than $2^{g(m)}$ programs have length below $g(m)$, so the fraction outside $R^{\log}$ at length $m$ is $O(m^{-d})$. A fixed-deficiency randomness predicate need not have density tending to one, which is why the density property below uses logarithmic deficiency.

For true $x{\in}R^{\log}$, let $H_{\mathcal S}(x)$ mean that no polynomial bounds the $\mathcal S$-proof lengths of $Con_{\mathcal S{+}(x{\in}R^{\log})}(n)$. This fixes $x$ before varying $n$. Truth is essential: a false randomness axiom makes its extension inconsistent, so a sound $\mathcal S$ eventually has no proofs of its finite consistency, giving vacuous hardness.

Fix a compatible universal conditional machine. For $|x|{=}|y|{=}m$, say that $x$ remains random conditional on $y$ when $K_U(x{\mid}y){\ge}g(m)$. Write $x{\mathrel{\perp_{R^{\log}}}}y$ when $x,y{\in}R^{\log}$ and each remains random conditional on the other. For the empty string $\varepsilon$, set $K_U(x{\mid}\varepsilon){:=}K_U(x)$ and let $\Delta(\varepsilon)$ add no axiom; for a nonempty string $y$, let $\Delta(y)$ be $y{\in}R^{\log}$.

\begin{conjecture}\label{conj:kolmogorov-blindness}
(\textbf{Kolmogorov Blindness}) For every sound finitely axiomatized sequential $\mathcal S{\supseteq}S^1_2$ with its fixed proof presentation, there exist a positive rational $\eta_{\mathcal S}$ and a cutoff $k_R^{\mathcal S}$ such that for every $m{>}k_R^{\mathcal S}$ there exists $N_{R,m}^{\mathcal S}$ such that, for every true $x{\in}R^{\log}\cap\{0,1\}^m$, every $y$ that is either $\varepsilon$ or a string of length $m$ with $y{\in}R^{\log}$ true, and every $n{>}N_{R,m}^{\mathcal S}$,
\[
\mathcal S{+}\Delta(y){\sststile{}{2^{\eta_{\mathcal S}n}}}Con_{\mathcal S{+}(x{\in}R^{\log})}(n)\quad\Longrightarrow\quad K_U(x{\mid}y){<}g(m).
\]
\end{conjecture}

The rate $\eta_{\mathcal S}$ may depend on the fixed proof encoding, the machine $U$ and the deficiency coefficient $d$. The onset $N_{R,m}^{\mathcal S}$ is common to all strings of length $m$, but no small or computable bound for it is asserted. Encoding dependence is essential: repeating every proof bit $r$ times reduces the unpadded consistency parameter to approximately $n/r$, while exhaustive finite verification~(Pudl\'ak~\cite[Section~4.3]{Pudlak2020reflectionprinciples}) gives proofs of size $2^{Cn/r}$ up to polynomial factors, for a fixed $C$. Thus a rate arbitrarily close to one cannot hold for all such encodings. General polynomial recodings can also change an exponential scale to a stretched exponential one. We use the ordinary local encoding of Section~\ref{sec:preliminaries}; no implication from SETH for satisfiability is asserted.

\subsection{Hardness and Density Consequences}\label{sec:density-hard-sentences}\label{sec:pairwise-hardness}\label{sec:density-no-mutual-help}
Fix a sound finitely axiomatized sequential $\mathcal S{\supseteq}S^1_2$, with the machines and proof presentation above.

\begin{definition}[Density KH and Strong KH]\label{def:kolmogorov-hardness}\label{def:strong-kh}
\textbf{Density KH} holds if $\lim_{m\to\infty}2^{-m}|\{x{\in}\{0,1\}^m{:}x{\in}R^{\log}\text{ and }H_{\mathcal S}(x)\}|{=}1$. \textbf{Strong KH} holds if there is $k_{\mathrm{KH}}^{\mathcal S,U,d}$ such that every true $x{\in}R^{\log}$ with $|x|{>}k_{\mathrm{KH}}^{\mathcal S,U,d}$ satisfies $H_{\mathcal S}(x)$.
\end{definition}

Density KH permits a vanishing fraction of easy true axioms and specifies no convergence rate or computable onset. The next theorem derives both properties and the single-axiom and pairwise finite exponential bounds from Blindness. The pairwise target is the consistency of $\mathcal S{+}(x{\in}R^{\log})$, not of the joint extension; the assertion allows proof shortening above the prescribed lower bound.
\begin{theorem}[Hardness consequences of Blindness]\label{thm:blindness-consequences}
Assume Kolmogorov Blindness at $\mathcal S$. Let $m{>}k_R^{\mathcal S}$ and $n{>}N_{R,m}^{\mathcal S}$.
\begin{enumerate}
\item For every true $x{\in}R^{\log}\cap\{0,1\}^m$, $\mathcal S\centernot{\sststile{}{2^{\eta_{\mathcal S}n}}}Con_{\mathcal S{+}(x{\in}R^{\log})}(n)$. At least $(1{-}O(m^{-d}))2^m$ strings of length $m$ satisfy membership and this bound.
\item For $x,y{\in}\{0,1\}^m$ with $x{\mathrel{\perp_{R^{\log}}}}y$, $\mathcal S{+}(y{\in}R^{\log})\centernot{\sststile{}{2^{\eta_{\mathcal S}n}}}\allowbreak Con_{\mathcal S{+}(x{\in}R^{\log})}(n)$. The same bound holds with $x,y$ interchanged. At least $(1{-}O(m^{-d}))2^{2m}$ ordered pairs satisfy the mutual randomness conditions and both bounds for all sufficiently large such $m$.
\item Strong KH and Density KH hold. One may take $k_{\mathrm{KH}}^{\mathcal S,U,d}{:=}k_R^{\mathcal S}$.
\end{enumerate}
\end{theorem}
\begin{proof}
For (1), take $y{=}\varepsilon$ in Blindness. A proof within the stated bound would give $K_U(x){<}g(m)$, contradicting $x{\in}R^{\log}$. Counting programs shorter than $g(m)$ bounds the exceptional fraction by $O(m^{-d})$.

For (2), apply Blindness to $(x,y)$ and $(y,x)$, using the two conditional incompressibility bounds. For each fixed $y$, fewer than $2^{g(m)}$ strings $x$ violate $K_U(x{\mid}y){\ge}g(m)$. The same estimate holds in the opposite direction. The two unconditional membership failures also affect only an $O(m^{-d})$ fraction of ordered pairs. A union bound proves the density estimate. Symmetry supplies both hardness bounds with the same onset.

For (3), fix a true $x{\in}R^{\log}$ above the threshold. Its eventual lower bound $2^{\eta_{\mathcal S}n}$ exceeds every polynomial because $\eta_{\mathcal S}{>}0$, so Strong KH holds. Counting the true members of $R^{\log}$ then gives Density KH.
\end{proof}

These density conclusions use the same strings or pairs for every $n$ beyond the common onset. The encoding of a randomness assertion uses the fixed-template observation of Calude and J\"urgensen~\cite[proof of Theorem~5.2]{CaludeJurgensen}: the assertion adds only constant overhead to the written representation of $x$. Here $x{\in}R^{\log}$ uses a fixed predicate that computes its threshold from $|x|$, so no separate threshold numeral is needed. By \emph{random tautologies} we mean the standard propositional consistency encodings indexed by $x{\in}R^{\log}$~(Kraj\'i\v{c}ek and Pudl\'ak~\cite[Section~3]{Krajicek}). Density is among these indices. The Cook encodings have size polynomial in $m{+}n$, and $S^1_2$ has proofs of polynomial size that their truth under every assignment implies the corresponding consistency statements. For a fixed Cook--Reckhow proof system $P$ whose soundness $\mathcal S$ proves, verifying a $P$-proof and applying $\mathrm{RFN}(P)$ therefore gives an $\mathcal S$-proof of the consistency statement with polynomial overhead. Blindness at such an $\mathcal S$ therefore yields hardness in $P$, with the lower bounds adjusted for this overhead. Universal Blindness allows $\mathcal S{:=}S^1_2{+}\mathrm{RFN}(P)$. The onsets may be large and need not be computable in $m$.

Universal Blindness therefore implies the universal Density KH property and, by Corollary~\ref{cor:optimal-characterization}, the nonexistence of a length-optimal propositional proof system and $\mathbf{NP}{\ne}\mathbf{coNP}$. This conclusion uses all eligible finite bases; a property at one fixed theory is not enough. No converse from Strong KH, Density KH or the single-axiom exponential property to Blindness is established. Applying a single-axiom assertion at $\mathcal S{+}(y{\in}R^{\log})$ changes the target to the joint extension and does not give the pairwise conclusion above.

\subsection{Factual Nonaccess and Certificate Obstructions}
Conditional incompressibility says that the bits of $y$ do not describe $x$ briefly; it does not by itself control the information in $Con_{\mathcal S{+}(y{\in}R^{\log})}$. Universal Blindness supplies a nonaccess consequence through an auxiliary finite theory, even when the original $\mathcal S$ has infinitely many polynomial-time decidable axioms.

\begin{definition}[Conditional Factual No Access property]\label{def:conditional-factual-no-access}
Fix a sound theory $\mathcal S{\supseteq}S^1_2$ with polynomial-time decidable axioms and the proof presentation of Section~\ref{sec:preliminaries}. We say that \textbf{Conditional Factual No Access} holds for $\mathcal S,R^{\log}$ if there exists $k_{R^{\log},\mathrm{access}}^{\mathcal S}$ such that, for every $m{>}k_{R^{\log},\mathrm{access}}^{\mathcal S}$ and every $x,y{\in}\{0,1\}^m$ satisfying $x{\mathrel{\perp_{R^{\log}}}}y$, one has $EA{+}Con_{\mathcal S{+}(y{\in}R^{\log})}{\not\vdash}x{\in}R^{\log}$ and $EA{+}Con_{\mathcal S{+}(x{\in}R^{\log})}{\not\vdash}y{\in}R^{\log}$.
\end{definition}

For the next lemma, let $R$ denote either $R^{\log}$ or fixed-deficiency randomness $K_U(x){\ge}|x|{-}c_U$, with its fixed presentation.

\begin{lemma}\label{lem:consistency-implies-randomness}
For each fixed $x$, $EA{\vdash}Con_{\mathcal S{+}(x{\in}R)}{\rightarrow}x{\in}R$. Consequently, for any sentence $\eta$, if $EA{\vdash}\eta{\rightarrow}Con_{\mathcal S{+}(x{\in}R)}$, then $EA{+}\eta{\vdash}x{\in}R$.
\end{lemma}

\begin{proof}
If $R{=}R^{\log}$ and $|x|{<}2$, its definition gives a fixed $S^1_2$-proof of nonmembership; combining it with the added axiom gives a refutation. Otherwise let $h_R(x)$ be the integer threshold in the randomness assertion. Falsity is witnessed by a program $p$ of length below $h_R(x)$ and a complete halting computation $h$ of $U(p)$ with output $x$. Checking the format and transitions of $h$, the output, and the length of $p$ is a polynomial-time computation. The computation-writing construction of Section~\ref{sec:preliminaries} gives an $S^1_2$-proof of the verified facts and hence, by existential introduction, of $\neg(x{\in}R)$. Its correctness is formalizable in $EA$, which can represent the computation and the resulting proof. Embed this proof into $\mathcal S{+}(x{\in}R)$ by the fixed proof embedding of Section~\ref{sec:preliminaries}, and use the added axiom to obtain a refutation. Thus $EA$ proves $\neg(x{\in}R){\rightarrow}\neg Con_{\mathcal S{+}(x{\in}R)}$, without assuming that $\mathcal S$ contains $EA$. Contraposition proves the first assertion; composition proves the second.
\end{proof}

Taking $\eta{=}Con_{\mathcal S}$, or $\eta{=}Con_{\mathcal S{+}(y{\in}R)}$, gives the single-axiom and pairwise certificate obstructions by contraposition. These implications alone give no proof-length lower bounds. Chaitin's theorem gives a finite cutoff for randomness facts provable in the sound theory $EA{+}Con_{\mathcal S}$, but no comparison with a hardness threshold is established.

\begin{proposition}[Universal Blindness implies factual nonaccess]\label{prop:universal-blindness-no-access}
Suppose \textbf{Kolmogorov Blindness} holds for every sound finitely axiomatized sequential theory extending $S^1_2$. Then every sound $\mathcal S{\supseteq}S^1_2$ with polynomial-time decidable axioms and the fixed proof presentation satisfies \textbf{Conditional Factual No Access} for $R^{\log}$. Hence, for all sufficiently long mutually random pairs, neither $EA{\vdash}Con_{\mathcal S{+}(y{\in}R^{\log})}{\rightarrow}Con_{\mathcal S{+}(x{\in}R^{\log})}$ nor the implication with $x,y$ interchanged holds.
\end{proposition}
\begin{proof}
Fix $\mathcal S$. Choose a fixed sound finitely axiomatized sequential base $A$ containing $EA$ and $S^1_2$ in the common definitional language, as in Section~\ref{sec:desiderata}. Using the uniform arithmetization of the axiom family, put $\chi_{\mathcal S}{:=}\forall z\,(z{\in}R^{\log}{\rightarrow}Con_{\mathcal S{+}(z{\in}R^{\log})})$ and $\mathcal T{:=}A{+}\chi_{\mathcal S}$. The sentence $\chi_{\mathcal S}$ is true because $\mathcal S$ is sound. Thus $\mathcal T$ is sound, finitely axiomatized and sequential, and satisfies Blindness by hypothesis.

Set $k_{R^{\log},\mathrm{access}}^{\mathcal S}{:=}k_R^{\mathcal T}$. Suppose $x{\mathrel{\perp_{R^{\log}}}}y$ has length above this threshold and $EA{+}Con_{\mathcal S{+}(y{\in}R^{\log})}{\vdash}(x{\in}R^{\log})$. The finite theory $Q{:=}\mathcal T{+}(y{\in}R^{\log})$ proves $Con_{\mathcal S{+}(y{\in}R^{\log})}$ using $\chi_{\mathcal S}$, and therefore proves $(x{\in}R^{\log})$. Replacing the added axiom $(x{\in}R^{\log})$ by this fixed $Q$-proof translates $(\mathcal T{+}(x{\in}R^{\log}))$-refutations into $Q$-refutations with polynomial growth, verified in $S^1_2$. Since $Q$ is finitely axiomatized and sequential, its polynomial-size self-consistency proofs~(Pudl\'ak~\cite[Theorem~3.1(2)]{Pudlak2017Author}) then give polynomial-size $Q$-proofs of $Con_{\mathcal T{+}(x{\in}R^{\log})}(n)$. This contradicts Theorem~\ref{thm:blindness-consequences}(2) for the ambient theory $\mathcal T$, since its positive exponential rate exceeds every polynomial for sufficiently large $n$. Interchanging $x,y$ proves the other nonaccess assertion. The relative-consistency conclusions follow by contraposition of Lemma~\ref{lem:consistency-implies-randomness}.
\end{proof}

This implication uses Blindness for the auxiliary theory, not just at the original $\mathcal S$. Conversely, factual nonaccess alone gives certificate obstructions, with no established implication to exponential hardness. At a fixed $\mathcal S$, Blindness also excludes $S^1_2$-proofs of $Con_{\mathcal S{+}(y{\in}R^{\log})}{\rightarrow}Con_{\mathcal S{+}(y{\in}R^{\log}){+}(x{\in}R^{\log})}$ for mutually random $x,y$ above its threshold: Theorem~\ref{thm:relative-consistency-implies-simulation} would give short consistency proofs for the joint extension, hence for $\mathcal S{+}(x{\in}R^{\log})$, contradicting the pairwise bound.

Changes of randomness convention preserve KH under suitable provable translations. Fix a sound finitely axiomatized sequential $\mathcal S{\supseteq}S^1_2$ with its standard proof presentation. If $\mathcal S{\vdash}\beta_x{\rightarrow}\alpha_x$ for each sufficiently long $x$, non-simulation of $\mathcal S{+}\alpha_x$ implies non-simulation of $\mathcal S{+}\beta_x$: replacing uses of $\alpha_x$ by its fixed derivation from $\beta_x$ gives a polynomial proof translation. Thus Strong KH transfers to stronger true axioms, and Density KH transfers when the stronger axioms are true for a fraction of strings tending to one; provably equivalent axiom families have equivalent KH properties. Standard comparisons between optimal plain machines and between plain and prefix-free complexity require constant or logarithmic threshold adjustments (Shen~\cite[Sections~3--4 and~35]{Shen2015Kolmogorov}). When $\mathcal S$ verifies the corresponding program translations, the schema asserting either KH property for every positive rational logarithmic-deficiency coefficient is therefore invariant under these changes. This does not establish equivalence at an identical fixed threshold or extend automatically to conditional, time-bounded or oracle-relative randomness; finite-scale bounds additionally require proof translations uniform in $x$.

\section{Effective Density and Randomness}\label{sec:effective-density}\label{sec:further-extensions}

We prove a quantitative equivalence between dense hardness and hardness above a corresponding randomness threshold, for fixed $\Pi_1$ axiom families with computable proof-length bounds and a computable common onset. Counting gives one direction; in the other, an exception's position in an effective enumeration gives a short description. The coding argument is the standard incompressibility method~(Shen~\cite[Theorems~2(d) and~9]{Shen2015Kolmogorov}); its application here covers whole hard tails and pairs whose bounds hold after adjoining the other axiom. All complexities below are plain complexities for the fixed optimal machines.

\begin{theorem}[Effective density and incompressibility]\label{thm:effective-density-characterization}
Let $L,f{:}\mathbb N{\to}\mathbb N$ be total computable functions with $0{\le}f(m){\le}L(m)$, and let $G_m{\subseteq}\{0,1\}^{L(m)}$. Suppose the complements $E_m{:=}\{0,1\}^{L(m)}\setminus G_m$ are uniformly computably enumerable from $m$. The following are equivalent:
\begin{enumerate}
\item There is a constant $C$ such that $|E_m|{\le}C2^{L(m)-f(m)}$ for all sufficiently large $m$.
\item There is an integer constant $c$ such that, for all sufficiently large $m$ and every $z{\in}\{0,1\}^{L(m)}$, $K_U(z{\mid}m){\ge}L(m){-}f(m){+}c\quad\Longrightarrow\quad z{\in}G_m$.
\end{enumerate}
\end{theorem}

\begin{proof}
For $(1){\Rightarrow}(2)$ choose an integer $b{\ge}0$ with $C{\le}2^b$. Given $m$, enumerate $E_m$ without repetition. An element of $E_m$ is described by its zero-based position in this enumeration, using exactly $L(m){-}f(m){+}b$ bits. The decoder computes this length from $m$ and outputs the element at that position when it appears. It need not detect when the enumeration is complete. Optimality of the fixed conditional machine gives $z{\in}E_m\quad\Longrightarrow\quad K_U(z{\mid}m){\le}L(m){-}f(m){+}b{+}a$ for a fixed decoder constant $a$. Take $c{=}b{+}a{+}1$. Conversely, (2) puts every exception below complexity $L(m){-}f(m){+}c$. Fewer than $2^{L(m)-f(m)+c}$ strings have such descriptions, proving (1). This reverse direction does not require enumerability.
\end{proof}

The loss from conditioning on $m$ can be avoided when the description length determines $m$ up to bounded ambiguity.

\begin{proposition}[Unconditional form]\label{prop:effective-density-unconditional}
Let $L$, $f$, $G_m$ and $E_m$ be as in Theorem~\ref{thm:effective-density-characterization}, and suppose that there are integers $r$ and $m'$ such that at most $r$ integers $m{\ge}m'$ share each value of $L(m){-}f(m)$. Then condition~(1) of Theorem~\ref{thm:effective-density-characterization} is equivalent to:
\begin{enumerate}
\item[(3)] There is an integer constant $c$ such that, for all sufficiently large $m$ and every $z{\in}\{0,1\}^{L(m)}$, $K_U(z){\ge}L(m){-}f(m){+}c\quad\Longrightarrow\quad z{\in}G_m$.
\end{enumerate}
\end{proposition}

\begin{proof}
Assume (1) for $m{\ge}m_0$ with constant $C$, where we may take $m_0{\ge}m'$, and choose an integer $b{\ge}0$ with $rC{\le}2^b$. For each $j$ let $V_j$ be the union of the sets $E_m$ with $m{\ge}m_0$ and $L(m){-}f(m){+}b{=}j$. The relation $z{\in}V_j$ is enumerable, because $L$ and $f$ are computable and the $E_m$ are uniformly computably enumerable, and $|V_j|{\le}rC2^{j-b}{\le}2^j$. By~Shen~\cite[Theorem~2(d)]{Shen2015Kolmogorov} there is a constant $a$ with $K_U(z){\le}j{+}a$ for all $z{\in}V_j$. Hence every $z{\in}E_m$ with $m{\ge}m_0$ satisfies $K_U(z){\le}L(m){-}f(m){+}b{+}a$; take $c{=}b{+}a{+}1$. Conversely, (3) puts every exception below unconditional complexity $L(m){-}f(m){+}c$, and fewer than $2^{L(m)-f(m)+c}$ strings have such descriptions.
\end{proof}

For the hardness application, fix a sound, finitely axiomatized sequential $\mathcal S{\supseteq}S^1_2$ with its ordinary proof predicate, a fixed $\Pi_1$ formula $\psi(x)$, a total computable bound $b{:}\mathbb N{\to}\mathbb N$, and a total computable onset $H(m)$. Counterexamples to the truth of $\psi$ can be enumerated; this applies to the randomness predicates above. For example, $b(n){:=}\lfloor2^{\eta n}\rfloor$ for a fixed positive rational $\eta$ prescribes an exponential scale. The result also applies to polynomial or other computable bounds; it does not prove that the chosen bounds are hard. Let $A_m$ consist of the strings $x{\in}\{0,1\}^m$ such that $\psi(x)$ is true and $\forall n{>}H(m)\quad \ell_{\mathcal S}\bigl(Con_{\mathcal S{+}\psi(x)}(n)\bigr){>}b(n)$. Let $B_m$ consist of concatenations $xy$ with $x,y{\in}A_m$ for which, for every $n{>}H(m)$, both
\[
\begin{split}
\ell_{\mathcal S{+}\psi(y)}\bigl(Con_{\mathcal S{+}\psi(x)}(n)\bigr)&{>}b(n),\\
\ell_{\mathcal S{+}\psi(x)}\bigl(Con_{\mathcal S{+}\psi(y)}(n)\bigr)&{>}b(n).
\end{split}
\]
Thus $A_m$ requires truth and every bound after the common onset; $B_m$ additionally requires both bounds when the other axiom is adjoined. Concatenations have equal-length halves, and these sets are fixed before $n$ varies.

\begin{corollary}[Dense hard tails and random instances]\label{cor:effective-hardness-characterization}
For the fixed data just specified and any total computable $f$ with $0{\le}f(m){\le}m$, the following are equivalent:
\begin{enumerate}
\item The exceptional fractions for both $A_m$ and $B_m$ are $O(2^{-f(m)})$.
\item There are constants $c_1,c_2$ such that, for all sufficiently large $m$ and all $x,y{\in}\{0,1\}^m$,
\[
\begin{split}
K_U(x{\mid}m){\ge}m{-}f(m){+}c_1&\quad\Longrightarrow\quad x{\in}A_m,\\
K_U(xy{\mid}m){\ge}2m{-}f(m){+}c_2&\quad\Longrightarrow\quad xy{\in}B_m.
\end{split}
\]
\end{enumerate}
In particular the first implication gives the truth of $\psi(x)$, and the second gives the truth of both $\psi(x)$ and $\psi(y)$ together with both no-help bounds.
\end{corollary}

\begin{proof}
Both complements are uniformly computably enumerable from $m$. Write $\psi(x)$ in the form $\forall z\,R(x,z)$ with decidable $R$. For each string of length $m$, search for a counterexample to $R(x,z)$ and, in parallel, search over integers $n{>}H(m)$ and candidate $\mathcal S$-proofs of $Con_{\mathcal S{+}\psi(x)}(n)$ of length at most $b(n)$. Output the string when either search finds a witness. The functions $H$ and $b$ are computable, and the ordinary proof predicate is decidable uniformly in $x$, so every exception is eventually listed and no other string is listed. For pairs, run the corresponding searches for both truth assertions, both base-theory bounds and both bounds after adjoining the other axiom. These searches need not be efficient or recognize when an enumeration is complete. Apply Theorem~\ref{thm:effective-density-characterization} with $L(m){=}m$ and $L(m){=}2m$, respectively.
\end{proof}

The bound for $A_m$ in~(1) follows from the bound for $B_m$ with the same $f$: each $x{\notin}A_m$ gives $2^m$ concatenations $xy{\notin}B_m$, so $|\{0,1\}^m\setminus A_m|{\le}2^{-m}|\{0,1\}^{2m}\setminus B_m|$. Condition~(1) may therefore be stated for $B_m$ alone.

Enumerability follows from the hypotheses, but sparsity remains a premise. For randomness axioms, counting bounds the false instances by $O(m^{-d})$; the fraction of true instances violating a hardness bound is unknown. For example, $\psi(x){\equiv}0{=}0$ and $b(n){=}2^{\lfloor n/2\rfloor}$ make every string exceptional, since short self-consistency proofs give polynomial-size proofs for this trivial extension.

\paragraph{Polynomially small exceptions.}
An exceptional fraction $O(m^{-a})$, for positive rational $a$, is equivalent to coverage of every sufficiently long $x$ with $K_U(x){\ge}m{-}a\log_2m{+}c$, for a suitable constant $c$. For pairs the threshold is $K_U(xy){\ge}2m{-}a\log_2m{+}c'$. Each equivalence applies separately to $A_m$ or $B_m$.

To obtain unconditional complexity, apply Proposition~\ref{prop:effective-density-unconditional} with $f(m){=}\min\{m,\lfloor a\log_2m\rfloor\}$. For sufficiently large $m$, $f(m{+}j){-}f(m){\le}1$ for $j{\in}\{1,2\}$, so $m{-}f(m)$ is nondecreasing with fibers of size at most two, while $2m{-}f(m)$ is strictly increasing. Both meet the multiplicity hypothesis. In particular, for $\delta{<}a$, all sufficiently long strings with $K_U(x){\ge}m{-}\delta\log_2m$ lie in $A_m$, and the analogous joint threshold gives membership in $B_m$. Joint incompressibility is distinct from the separate conditions $x{\mathrel{\perp_{R^{\log}}}}y$. For randomness axioms the attainable rate has the following limitation.

\begin{proposition}[Limitation of the density rate]\label{prop:density-rate-limitation}
If $\{x{\in}\{0,1\}^m{:}x{\notin}R^{\log}\}$ has density $O(m^{-a})$ for a real $a{>}0$, then $a{\le}d$. Consequently the preceding density argument, with $\psi(x){:=}x{\in}R^{\log}$, guarantees coverage by this rate comparison for $\delta{<}a{\le}d$; the endpoint $\delta{=}d$ and Strong KH are not supplied by the rate comparison alone.
\end{proposition}
\begin{proof}
Suppose $a{>}d$ and choose a rational $b$ with $d{<}b{<}a$. For $j{\ge}2$, a fixed machine maps each $j$-bit string $y$ to $x{=}y0^{\lceil b\log_2j\rceil}$, of length $m_j{:=}j{+}\lceil b\log_2j\rceil$. For a constant $c_b$, plain complexity satisfies $K_U(x){\le}j{+}c_b$: the description is $y$, whose length tells the decoder how much padding to add~(Shen~\cite[Theorem~2(c)]{Shen2015Kolmogorov}). But $g(m_j){-}j{=}(b{-}d)\log_2j{+}O(1){\longrightarrow}\infty$. Thus, for all sufficiently large $j$, all $2^j$ distinct outputs lie outside $R^{\log}$. Their fraction at length $m_j$ is $2^{-\lceil b\log_2j\rceil}{=}\Omega(m_j^{-b})$, contradicting an $O(m_j^{-a})$ bound because $a{>}b$. For the consequence, every $x$ outside $R^{\log}$ also lies outside $A_m$.
\end{proof}

The equivalence supplies neither the density estimate nor a computable onset. Blindness supplies common onsets but does not assert they are computably bounded; density separately at each $n$ need not give a dense set satisfying all $n{>}H(m)$. Density KH also specifies no polynomial rate, so the criterion does not turn it into Strong KH. Effectiveness is essential: without it, even density $1{-}2^{-m}$ need not cover all incompressible strings, as shown by omitting one string of complexity at least $m$ at each length.

\section{Remarks on Provability}\label{sec:provability}
Corollary~\ref{cor:optimal-characterization} characterizes the nonexistence of a length-optimal proof system by universal Busy Beaver non-simulation. What can theories prove about these assertions and Kolmogorov hardness? Non-simulation requires lower-bound instances beyond every onset for every polynomial exponent; Blindness asserts the stronger eventual exponential bounds.

The familiar self-unprovability argument concerns uniform assertions. It cannot prohibit individual lower-bound proofs: Hrube\v{s}'s lemma gives polynomial-size proofs, for each standard $n$, that a chosen non-theorem has no proof shorter than $\log_2n$, including when the proving theory is the theory whose proofs are bounded. The following proposition makes the uniform restriction precise.

\begin{proposition}\label{prop:no-self-certification}
Let $\mathcal T{\supseteq}EA$ be consistent, with polynomial-time decidable axioms.
\begin{enumerate}
\item Let $\mathrm{Blind}(\mathcal T)$ be the sentence $\exists b\,\forall x\,\neg Prov_{\mathcal T}(\ulcorner K_U(\bar x){>}\bar b\urcorner)$. Then $EA{\vdash}\mathrm{Blind}(\mathcal T){\rightarrow}Con_{\mathcal T}$, and therefore $\mathcal T{\not\vdash}\mathrm{Blind}(\mathcal T)$. If $\mathcal T$ is sound, $\mathrm{Blind}(\mathcal T)$ is true.
\item Let $f$ be a clocked polynomial-time function and $r$ a polynomial such that $EA$ proves that, for every $n$, $f(n)$ is the code of a sentence of size at most $r(n)$. Let $\mathrm{NonSim}(\mathcal T,f)$ be the arithmetized statement that for every $c$ and every $N$ there is $n{\ge}N$ such that the sentence coded by $f(n)$ has no $\mathcal T$-proof of size at most $n^c$. Then $EA{\vdash}\mathrm{NonSim}(\mathcal T,f){\rightarrow}Con_{\mathcal T}$, and therefore $\mathcal T{\not\vdash}\mathrm{NonSim}(\mathcal T,f)$.
\end{enumerate}
\end{proposition}

\begin{proof}
Both parts are explosion arguments formalized in $EA$. For part (1), a $\mathcal T$-proof of $0{=}1$ extends to a $\mathcal T$-proof of $K_U(\bar 0){>}\bar b$ for every $b$, so $\neg Con_{\mathcal T}$ refutes $\mathrm{Blind}(\mathcal T)$. Truth for sound $\mathcal T$ is Chaitin's theorem~\cite{ChaitinIncompleteness}; see also~Shen~\cite[Theorem~4]{Shen2015Kolmogorov}. For part (2), a $\mathcal T$-proof $w$ of $0{=}1$ extends, for every $n$, to a $\mathcal T$-proof of the sentence coded by $f(n)$ of size at most $q(|w|{+}n)$ for a fixed polynomial $q$ bounding proof composition and the explosion step. The theory $EA$ verifies this construction because it proves that $f(n)$ codes a sentence of size at most $r(n)$. For fixed $w$, choose $c$ and $N$ such that $q(|w|{+}n){\le}n^c$ for every $n{\ge}N$; then every such $n$ has a proof of size at most $n^c$, so $\neg Con_{\mathcal T}$ refutes $\mathrm{NonSim}(\mathcal T,f)$. In both parts the second incompleteness theorem for $\mathcal T$ gives the unprovability.
\end{proof}

For $\mathcal T{=}\mathcal S{\supseteq}EA$, the function $f(n){=}\ulcorner Con_{\mathcal S{+}\phi}(\bar n)\urcorner$ satisfies part (2)'s internal well-formedness hypothesis by fixed-formula numeral substitution. Thus $\mathcal S$ cannot prove its own uniform non-simulation assertion for a fixed extension, including any fixed randomness axiom. The same holds for $f(n){=}\ulcorner Con_{S^1_2{+}\mathsf{Exp}{+}\phi_{BB}(k)}(\bar n)\urcorner$ at each fixed $k$. This leaves open provability in a stronger theory. The classical, non-quantitative sentence $\mathrm{Blind}(\mathcal T)$ in part (1) is distinct from Conjecture~\ref{conj:kolmogorov-blindness}.

A stronger ambient theory nevertheless has its own randomness limitations. For a sound effectively axiomatized $\mathcal S$, let $A_i$ be the set of true $\Pi_i$ sentences and $\mathcal T_i{:=}\mathcal S{+}A_i$, with $A_0{:=}\emptyset$. Its theorems are enumerable relative to $A_i$. Chaitin's search argument therefore shows that $\mathcal T_i$ proves only finitely many assertions $K^{A_i}(x){\ge}g(|x|)$: otherwise an oracle program given $L$ could find such an $x$ of length above $L$ with description length $O(\log L)$, contradicting soundness. Counting still gives density one for these oracle-random strings. This suggests an obstacle to proving hardness for a weaker theory, but proves no such unprovability: inability to establish individual randomness facts does not exclude a uniform theorem about the weaker theory's proof lengths.

Self-unprovability does not establish that the universal Busy Beaver property, or equivalently the nonexistence of a length-optimal proof system, is unprovable in a specified ambient theory such as $\mathsf{PA}$. Nor is inability to prove unprovability a distinctive feature of these hardness assertions. For any $\mathcal T$ as in Proposition~\ref{prop:no-self-certification} and any fixed sentence $\theta$, formalized explosion gives $EA{\vdash}\neg Prov_{\mathcal T}(\ulcorner\theta\urcorner){\rightarrow}Con_{\mathcal T}$. Hence $\mathcal T{\not\vdash}\neg Prov_{\mathcal T}(\ulcorner\theta\urcorner)$ by the second incompleteness theorem (Artemov and Beklemishev~\cite[Theorem~19]{ArtemovBeklemishev2005}), regardless of the subject matter of $\theta$.

\paragraph{Unifying hardness criteria.}
These characterizations parallel connections between circuit size and $KT$ complexity~(Allender et al.~\cite[Theorem~11]{ABKMR}, with the authors' clarification~\cite{Clarification}) and between one-way functions and the average-case hardness of time-bounded Kolmogorov complexity $K^t$~(Liu and Pass~\cite[Theorem~1.1]{LP}). It would be desirable to formulate a single criterion in terms of information constraints whose instances imply proof-length lower bounds, circuit lower bounds, and average-case hardness of inverting candidate one-way functions. Such a criterion would need precise resource bounds and, for average-case claims, specified distributions and success probabilities.

A criterion based on the absence of an efficiency certificate encounters a proof-theoretic obstacle. Fix a certificate theory $B$ satisfying the hypotheses of Proposition~\ref{prop:no-self-certification}, and let $E_\tau$ and $H_\tau$ be arithmetical assertions of efficiency and hardness for a precisely specified task $\tau$. Consider a proposed rule $\neg Prov_B(\ulcorner E_\tau\urcorner){\rightarrow}H_\tau$. In any ambient theory proving this rule, establishing its premise already suffices to prove the desired hardness. If that ambient theory is $B$, the preceding argument prevents it from proving the premise; a stronger theory may do so. Mechanically formulating such a schema remains possible, but its usefulness depends on how its instances are justified and their premises established. The refutation of HRC in Section~\ref{sec:desiderata} gives a concrete caution: failure of a specified relative-consistency certificate need not imply non-simulation. A proposed unification must therefore justify both the passage from a missing certificate to hardness and a way to establish that the certificate is missing.

\section{Conclusion}\label{sec:conclusion}
The paper's main structural result characterizes the nonexistence of a length-optimal propositional proof system by universal Busy Beaver non-simulation (Corollary~\ref{cor:optimal-characterization}). It combines Busy Beaver transfer with the Kraj\'i\v{c}ek--Pudl\'ak correspondence. The transfer equates, at each fixed theory, the existence of a hard true extension with eventual Busy Beaver non-simulation; finite lower bounds transfer with explicit changes of scale. Relative consistency provable in $S^1_2$ yields simulation when the proving theory has short self-consistency proofs. The effective density theorem characterizes coverage of incompressible indices by specified true hard tails and pairs, under computable bounds and a common computable onset.

Kolmogorov Blindness proposes exponential hardness for sufficiently long randomness axioms, including after adjoining another true random axiom when conditional randomness is preserved. It implies Strong KH, Density KH and their finite-scale forms; universally, it also implies factual nonaccess. The density theorem does not prove the sparsity premise or the computable onset needed to derive these hardness assertions. Pudl\'ak's consistency-strength conjecture remains separate.

Hrube\v{s}'s construction refutes the earlier characterization by weak-base relative consistency. Whether ordinary $EA$ relative consistency suffices for simulation remains open here, as does a characterization of individual true extensions.

\begingroup
\small
\bibliographystyle{amsplain}
\bibliography{equivalence}
\endgroup

\end{document}